\documentclass[preprint,3p,times]{elsarticle}
\RequirePackage{multirow,booktabs,subfigure,color,array,hhline,makecell}
\usepackage{amssymb}
\usepackage{amsmath}
\usepackage{graphicx}
\usepackage{amsthm}
\usepackage{mathrsfs}
\usepackage{indentfirst}
\usepackage[colorlinks,citecolor=blue,urlcolor=blue]{hyperref}
\allowdisplaybreaks
\usepackage[table]{xcolor}
\usepackage{tikz-network}
\usepackage{pgf}
\usepackage{tikz}
\usepackage{subfigure}

\usepackage{bbm}

\usetikzlibrary{arrows, decorations.pathmorphing, backgrounds, positioning, fit, petri, automata}
\usetikzlibrary{shadows,arrows,positioning}
\RequirePackage{algorithm,algpseudocode}
\allowdisplaybreaks
\usepackage{changes}

\usepackage{setspace}
\newtheorem{thm}{Theorem}
\newtheorem{defin}{Definition}
\newtheorem{lem}{Lemma}
\newtheorem{assum}{Assumption}
\newtheorem{rem}{Remark}

\newtheorem{con}{Condition}

\allowdisplaybreaks[4]
\AtBeginDocument{%
	\providecommand\BibTeX{{%
			\normalfont B\kern-0.5em{\scshape i\kern-0.25em b}\kern-0.8em\TeX}}}
\journal{~}
\begin{document}
\begin{frontmatter}
\title{Estimating mixed memberships in multi-layer networks}
\author[label1]{Huan Qing\corref{cor1}}
\ead{qinghuan@u.nus.edu;qinghuan@cqut.edu.cn;qinghuan07131995@163.com}
\cortext[cor1]{Corresponding author.}
\address[label1]{School of Economics and Finance, Lab of Financial Risk Intelligent Early Warning and Modern Governance,  Chongqing University of Technology, Chongqing, 400054, China}
\begin{abstract}
Community detection in multi-layer networks has emerged as a crucial area of modern network analysis. However, conventional approaches often assume that nodes belong exclusively to a single community, which fails to capture the complex structure of real-world networks where nodes may belong to multiple communities simultaneously. To address this limitation, we propose novel spectral methods to estimate the common mixed memberships in the multi-layer mixed membership stochastic block model. The proposed methods leverage the eigen-decomposition of three aggregate matrices: the sum of adjacency matrices, the debiased sum of squared adjacency matrices, and the sum of squared adjacency matrices. We establish rigorous theoretical guarantees for the consistency of our methods. Specifically, we derive per-node error rates under mild conditions on network sparsity, demonstrating their consistency as the number of nodes and/or layers increases under the multi-layer mixed membership stochastic block model. Our theoretical results reveal that the method leveraging the sum of adjacency matrices generally performs poorer than the other two methods for mixed membership estimation in multi-layer networks. We conduct extensive numerical experiments to empirically validate our theoretical findings. For real-world multi-layer networks with unknown community information, we introduce two novel modularity metrics to quantify the quality of mixed membership community detection. Finally, we demonstrate the practical applications of our algorithms and modularity metrics by applying them to real-world multi-layer networks, demonstrating their effectiveness in extracting meaningful community structures.
\end{abstract}
\end{frontmatter}
\section{Introduction}\label{sec1}
Multi-layer networks have emerged as a powerful tool for describing complex real-world systems. Unlike single-layer networks, these structures consist of multiple layers, where nodes represent entities and edges represent their interactions \citep{mucha2010community,kivela2014multilayer,boccaletti2014structure}. In this paper, we focus on multi-layer networks with the same nodes set in all layers and nodes solely interacting within their respective layers, excluding any cross-layer connections. Such networks are ubiquitous, spanning from social networks to transportation systems, biological networks, and trade networks. For instance, in social networks, individuals tend to form connections within different platforms (e.g., Facebook, Twitter, WeChat, LinkedIn), while cross-platform connections are typically not allowed. In transportation networks, layers might represent different modes of transportation (roads, railways, airways), with nodes corresponding to locations and edges indicating the availability of a specific mode between two locations \citep{boccaletti2014structure}. The absence of cross-layer connections indicates the inability to directly switch modes without intermediate stops. Biological networks also exhibit rich multi-layer structures. In gene regulatory networks, layers could represent different types of gene interactions, with nodes representing genes and edges depicting their specific interactions \citep{narayanan2010simultaneous,bakken2016comprehensive,zhang2017finding}. The absence of cross-layer connections underscores the specialized nature of interactions within each layer and their vital role in governing cellular functions. The FAO-trade multi-layer network gathers trade relationships between countries across various products sourced from the Food and Agriculture Organization of the United Nations \citep{de2015structural}. Figure \ref{FAOTRADE3N} illustrates the networks of the first three products within this dataset.
\begin{figure}
\centering
\resizebox{\columnwidth}{!}{
\subfigure[]{\includegraphics[width=0.2\textwidth]{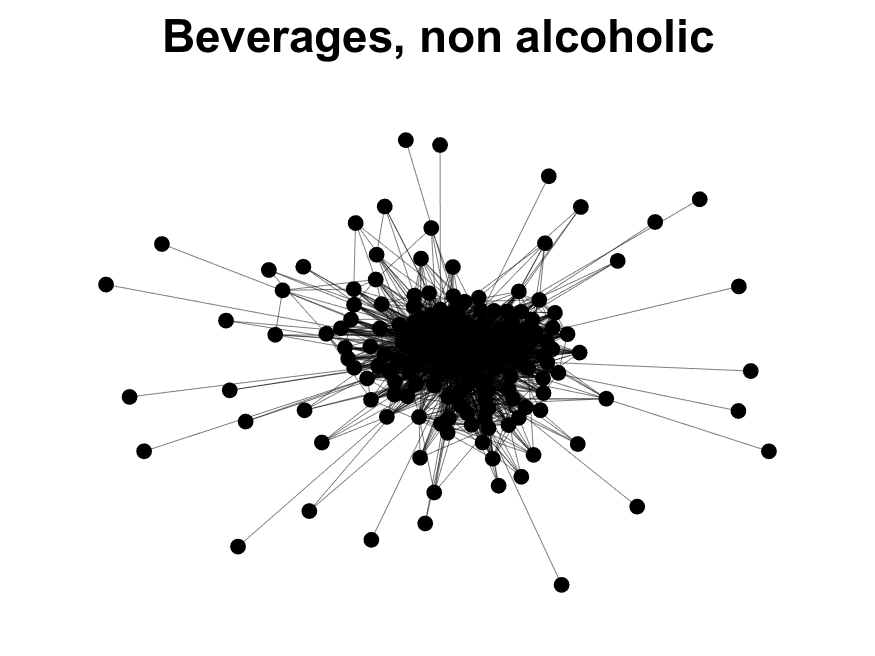}}
\subfigure[]{\includegraphics[width=0.2\textwidth]{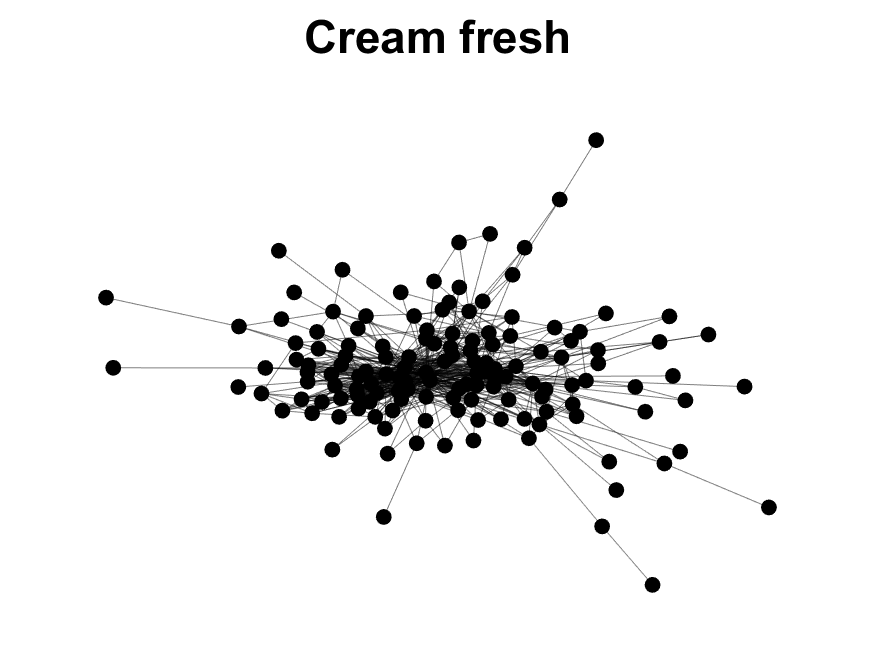}}
\subfigure[]{\includegraphics[width=0.2\textwidth]{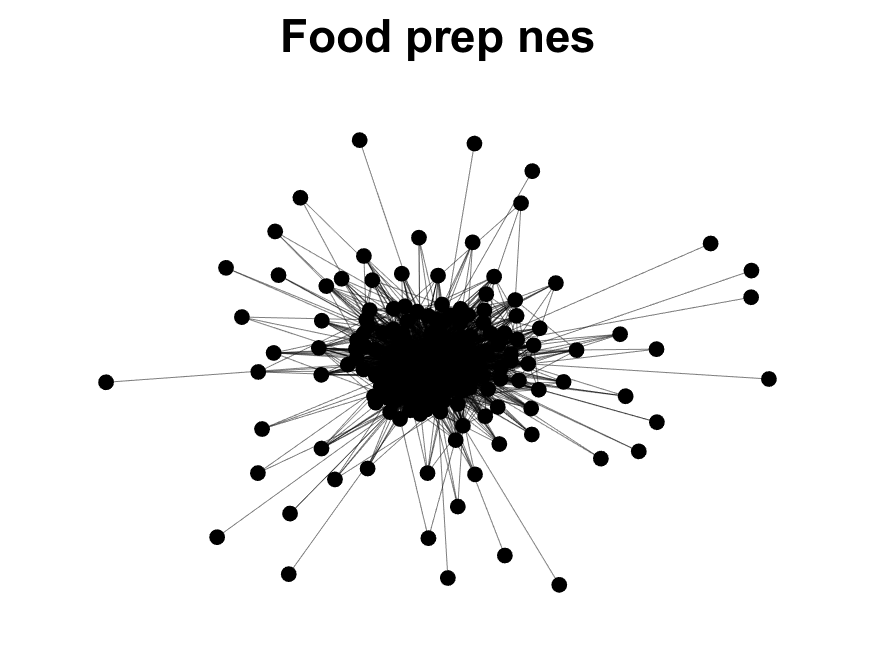}}
}
\caption{Networks of the first three products of the FAO-trade multi-layer network.}
\label{FAOTRADE3N} 
\end{figure}

Community detection in multi-layer networks is a crucial analytical tool, revealing latent structures within networks \citep{kim2015community, huang2021survey}. A community (also known as a group, cluster, or block) typically comprises nodes more densely interconnected than those outside \citep{newman2003structure,newman2004finding,fortunato2010community,fortunato2016community,javed2018community}. In social networks, individuals often form distinct communities based on interests, occupations, or locations. For instance, on Facebook, users might belong to hobby-based communities like photography or hiking, identifiable by dense interaction patterns among members. In transportation networks, communities emerge due to geographical proximity or functional similarity. For instance, cities might form communities linked by tightly integrated railways. In practical applications, a node can simultaneously belong to multiple communities. For example, in social networks, an individual may be a member of several different social groups. Similarly, in transportation networks, a specific location can function as a hub for multiple modes of transportation, bridging diverse communities. Likewise, in biological networks, a gene can belong to several overlapping communities, participating in a wide range of processes.

In the last few years, community detection in multi-layer networks, where each node belongs exclusively to a single community, has attracted considerable attention. For instance, several studies have been proposed under the multi-layer stochastic block model (MLSBM), a model that assumes the network for each layer can be modeled by the well-known stochastic block model (SBM) \citep{holland1983stochastic}. MLSBM also assumes that the community information is common to all layers. Under the framework within MLSBM, \citep{han2015consistent} studied the consistent community detection for maximum likelihood estimate and a spectral method designed by running the K-means algorithm on a few eigenvectors of the sum of adjacency matrices when only the number of layers grows. \citep{paul2020spectral} studied the consistency of co-regularized spectral clustering, orthogonal linked matrix factorization, and the spectral method in \citep{han2015consistent} as both the number of nodes and the number of layers grow within the MLSBM context. Moreover, \citep{lei2020consistent} established consistency results for a least squares estimation of community memberships within the MLSBM framework. Recently, \citep{lei2023bias} studied the consistency of a bias-adjusted (i.e., debiased) spectral clustering method using a novel debiased sum of squared adjacency matrices under MLSBM. Their numerical experiments demonstrated that this approach significantly outperforms the spectral method proposed by \citep{han2015consistent}. Also see \citep{pensky2019spectral,jing2021community,chen2022community,su2023spectral} for other recent works exploring variants of MLSBM.

However, the MLSBM model has one significant limitation, being tailored exclusively for multi-layer networks with non-overlapping communities. The mixed membership stochastic block (MMSB) model \citep{MMSB} is a popular statistical model for capturing mixed community memberships in single-layer networks. Under MMSB, some methods have been developed to estimate mixed memberships for single-layer networks generated from MMSB, including variational expectation-maximization \citep{MMSB, Gopalan2013eff}, nonnegative matrix factorization \citep{psorakis2011overlapping,wang2011community}, tensor decomposition \citep{anandkumar2014tensor}, and spectral clustering \citep{mao2021estimating,jin2024mixed}. In this paper, we consider the problem of estimating nodes' common community memberships in multi-layer networks generated from the multi-layer mixed membership stochastic block (MLMMSB) model, a multi-layer version of MMSB. The main contributions of this paper are summarized as follows:
\begin{itemize}
  \item We introduce three spectral methods for estimating mixed memberships in multi-layer networks generated from MLMMSB. These methods employ vertex-hunting algorithms on a few selected eigenvectors of three aggregate matrices: the sum of adjacency matrices, the debiased sum of squared adjacency matrices, and the sum of squared adjacency matrices.
  \item We establish per-node error rates for these three methods under mild conditions on network sparsity, demonstrating their consistent mixed membership estimation as the number of nodes and/or layers increases within the MLMMSB framework. Our theoretical analysis reveals that the method utilizing the debiased sum of squared adjacency matrices consistently outperforms the method using the sum of squared adjacency matrices in terms of error rate. Additionally, both methods generally exhibit lower error rates than the method based on the sum of adjacency matrices. This underscores the advantage of debiased spectral clustering in mixed membership community detection for multi-layer networks. To the best of our knowledge, this is the first work to estimate mixed memberships using the above three aggregate matrices and the first work to establish consistency results within the MLMMSB framework.
  \item To assess the quality of mixed membership community detection in multi-layer networks, we introduce two novel modularity metrics: fuzzy sum modularity and fuzzy mean modularity. The first is derived from computing the fuzzy modularity of the sum of adjacency matrices, while the second is obtained by averaging the fuzzy modularity of adjacency matrices across individual layers. To  the best of our knowledge, our two modularity metrics are the first to measure the quality of mixed membership community detection in multi-layer networks.
  \item We conduct extensive simulations to validate our theoretical findings and demonstrate the practical effectiveness of our methods and metrics through real-world multi-layer network applications.
\end{itemize}

The remainder of this paper is structured as follows: Section \ref{sec2} introduces the model, followed by Section \ref{sec3} outlining the methods developed. Section \ref{sec4} presents the theoretical results. Subsequently, Section \ref{sec5} includes experimental results on computer-generated multi-layer networks, while Section \ref{sec6realdata} focuses on real-world multi-layer networks. Lastly, Section \ref{sec7} concludes the paper and technical proofs are provided in \ref{SecProofs}.

\emph{Notation.} We employ the notation $[m]$ to denote the set $\{1,2,\ldots,m\}$. Furthermore, $I_{m\times m}$ stands for the $m$-by-$m$ identity matrix. For a vector $x$, $\|x\|_{q}$ represents its $l_{q}$ norm. When considering a matrix $M$ and an index set $s$ that is a subset of $[m]$, $M(s,:)$ refers to the sub-matrix of $M$ comprising the rows indexed by $s$. Additionally, $M'$ denotes the transpose of $M$, $\|M\|_{F}$ is its Frobenius norm, $\|M\|_{\infty}$ represents the maximum absolute row sum, $\|M\|_{2\rightarrow\infty}$ signifies the maximum row-wise $l_{2}$ norm, $\mathrm{rank}(M)$ gives its rank, and $\lambda_{k}(M)$ stands for the $k$-th largest eigenvalue of $M$ in magnitude. The notation $\mathbb{E}[\cdot]$ is used to denote expectation, while $\mathbb{P}(\cdot)$ represents probability. Finally, $e_{i}$ is defined as the indicator vector with a 1 in the $i$-th position and zeros elsewhere.
\section{Multi-layer mixed membership stochastic block model}\label{sec2}
Throughout, we consider undirected and unweighted multi-layer networks with $n$ common nodes and $L$ layers. For the $l$-th layer, let the $n\times n$ symmetric matrix $A_{l}$ be its adjacency matrix such that $A_{l}(i,j)=A_{l}(j,i)=1$ if there is an edge connecting node $i$ and node $j$ and $A_{l}(i,j)=A_{l}(j,i)=0$ otherwise, for all $i\in[n], j\in[n], l\in[L]$, i.e., $A_{l}=A'_{l}\in\{0,1\}^{n\times n}$. Additionally, we allow for the possibility of self-edges (loops) in this paper. We assume that the multi-layer network consists of $K$ common communities
\begin{align}\label{defineC}
\mathcal{C}_{1}, \mathcal{C}_{2}, \ldots, \mathcal{C}_{K}.
\end{align}

Throughout, we assume that the number of communities $K$ is known in this paper. The challenging task of theoretically estimating $K$ in multi-layer networks exceeds the scope of this paper. Let $\Pi\in[0,1]^{n\times K}$ be the membership matrix such that $\Pi(i,k)$ represents the ``weight'' that node $i$ belongs to the $k$-th community $\mathcal{C}_{k}$ for $i\in[n], k\in[K]$. Suppose that $\Pi$ satisfies the following condition
\begin{align}\label{definePi}
\mathrm{rank}(\Pi)=K, 0\leq\Pi(i,k)\leq1, \|\Pi(i,:)\|_{1}=1 \mathrm{~for~} i\in[n],k\in[K].
\end{align}

We call node $i$ a ``pure'' node if one entry of $\Pi(i,:)$ is 1 while the other $(K-1)$ elements are zeros and a ``mixed'' node otherwise. Assume that
\begin{align}\label{pureNodes}
\mathrm{The~}k\mathrm{-th~community~}\mathcal{C}_{k}~\mathrm{has~at~least~one~pure~node~for~} k\in[K].
\end{align}

Based on Equation (\ref{pureNodes}), let $\mathcal{I}$ be the index set of pure nodes such that $\mathcal{I}=\{p_{1},p_{2},\ldots,p_{K}\}$ with $p_{k}\in\{1,2,\ldots,n\}$ being an arbitrary pure node in the $k$-th community $\mathcal{C}_{k}$ for $k\in[K]$. Similar to \citep{mao2021estimating}, without loss of generality, we reorder the nodes such that $\Pi(\mathcal{I},:)=I_{K\times K}$.

Suppose that all the layers share a common mixed membership matrix $\Pi$ but with possibly different edge probabilities. In particular, we work with a multi-layer version of the popular mixed membership stochastic block model (MMSB) \citep{MMSB}. To be precise, we define the multi-layer MMSB below:
\begin{defin}
Suppose that Equations (\ref{defineC})-(\ref{pureNodes}) are satisfied, the multi-layer mixed membership stochastic block model (MLMMSB) for generating a multi-layer network with adjacency matrices $\{A_{l}\}^{L}_{l=1}$ is as follows:
\begin{align}\label{MLMMSB}
\Omega_{l}:=\rho\Pi B_{l}\Pi'~~~~A_{l}(i,j)=A_{l}(j,i)\sim \mathrm{Bernoulli}(\Omega_{l}(i,j))\mathrm{~for~}i\in[n],j\in[n], l\in[L],
\end{align}
where $B_{l}=B'_{l}\in[0,1]^{K\times K}$ for $l\in[L]$ and $\rho\in(0,1]$.
\end{defin}

Since $B_{l}$ can vary for different $l$, $A_{l}$ generated by Equation (\ref{MLMMSB}) may have different expectation $\Omega_{l}\equiv\mathbb{E}[A_{l}]$ for $l\in[L]$. Notably, when $L=1$, the MLMMSB model degenerates to the popular MMSB model. When all nodes are pure, MLMMSB reduces to the MLSBM studied in previous works \citep{han2015consistent,paul2020spectral,lei2020consistent,lei2023bias}. Define $\mathcal{B}=\{B_{1}, B_{2},\ldots,B_{L}\}$. Equation (\ref{MLMMSB}) says that MLMMSB is parameterized by $\Pi, \rho$, and $\mathcal{B}$. For brevity, we denote MLMMSB defined by Equation (\ref{MLMMSB}) as $\mathrm{MLMMSB}(\Pi,\rho,\mathcal{B})$. Since $\mathbb{P}(A_{l}(i,j)=1)=\Omega_{l}(i,j)=\rho\Pi(i,:)B_{l}\Pi'(j,:)\leq\rho$, we see that decreasing the value of $\rho$ results in a sparser multi-layer network, i.e., $\rho$ controls the overall sparsity of the multi-layer network. For this reason, we call $\rho$ the sparsity parameter in this paper. We allow the sparsity parameter $\rho$ to go to zero by either increasing the number of nodes $n$, the number of layers $L$, or both simultaneously. We will study the impact of $\rho$ on the performance of the proposed methods by considering $\rho$ in our theoretical analysis. The primal goal of community detection within the MLMMSB framework is to accurately estimate the common mixed membership matrix $\Pi$ from the observed $L$ adjacency matrices $\{A_{l}\}^{L}_{l=1}$. This estimation task is crucial for understanding the underlying community structure in multi-layer networks.
\section{Spectral methods for mixed membership estimation}\label{sec3}
In this section, we propose three spectral methods designed to estimate the mixed membership matrix $\Pi$ within the MLMMSB model for multi-layer networks in which nodes may belong to multiple communities with different weights. We recall that $\Omega_{l}$ is the expectation of the $l$-th observed adjacency matrix $A_{l}$ for $l\in[L]$ and $\Pi$ contains the common community memberships for all nodes. To provide intuitions of the designs of our methods, we consider the oracle case where $\{\Omega_{l}\}^{L}_{l=1}$ are directly observed. First, we define two distinct aggregate matrices formed by $\{\Omega_{l}\}^{L}_{l=1}$: $\Omega_{\mathrm{sum}}\equiv\sum_{l\in[L]}\Omega_{l}$ and $\tilde{S}_{\mathrm{sum}}\equiv\sum_{l\in[L]}\Omega^{2}_{l}$, where the former represents the sum of all expectation adjacency matrices, the latter is the sum of all squared expectation adjacency matrices, and both matrices contain the information about the mixed memberships of nodes in the multi-layer network since $\Omega_{\mathrm{sum}}=\rho\Pi(\sum_{l\in[L]}B_{l})\Pi'$ and $\tilde{S}_{\mathrm{sum}}=\rho^{2}\Pi(\sum_{l\in[L]}B_{l}\Pi'\Pi B_{l})\Pi'$. Since $\mathrm{rank}(\Pi)=K$, it is easy to see that $\mathrm{rank}(\Omega_{\mathrm{sum}})=K$ if $\mathrm{rank}(\sum_{l\in[L]}B_{l})=K$ and $\mathrm{rank}(\tilde{S}_{\mathrm{sum}})=K$ if $\mathrm{rank}(\sum_{l\in[L]}B^{2}_{l})=K$. Assuming that the number of communities $K$ is much smaller than the number of nodes $n$, we observe that $\Omega_{\mathrm{sum}}$ and $\tilde{S}_{\mathrm{sum}}$ possess low-dimensional structure. This low-rank property is crucial for the development of our spectral methods, as it allows us to efficiently extract meaningful information from the high-dimensional data. Building on these insights, we present the following lemma that characterizes the geometries of the two aggregate matrices $\Omega_{\mathrm{sum}}$ and $\tilde{S}_{\mathrm{sum}}$. This lemma forms the theoretical foundation for our methods, enabling us to develop accurate and computationally efficient algorithms for estimating the mixed membership matrix $\Pi$ within the MLMMSB model.
\begin{lem}\label{IS}
(Ideal Simplexes) Under the model $\mathrm{MLMMSB}(\Pi,\rho,\mathcal{B})$, depending on the conditions imposed on the set $\{B_{l}\}^{L}_{l=1}$, we arrive at the following conclusions:
\begin{enumerate}
  \item When $\mathrm{rank}(\sum_{l\in[L]}B_{l})=K$: Let $U\Sigma U'$ denote the top $K$ eigen-decomposition of $\Omega_{\mathrm{sum}}$ where the $n\times K$ matrix $U$ satisfies $U'U=I_{K\times K}$ and the $k$-th diagonal entry of the $K\times K$ diagonal matrix $\Sigma$ represents the $k$-th largest eigenvalue (in magnitude) of $\Omega_{\mathrm{sum}}$ for $k\in[K]$. Then, we have $U=\Pi U(\mathcal{I},:)$.

  \item When $\mathrm{rank}(\sum_{l\in[L]}B^{2}_{l})=K$: Let $V\Lambda V'$ represent the top $K$ eigen-decomposition of $\tilde{S}_{\mathrm{sum}}$ where the $n\times K$ matrix $V$ fulfills $V'V=I_{K\times K}$ and the $k$-th diagonal entry of the $K\times K$ diagonal matrix $\Lambda$ is the $k$-th largest eigenvalue (in magnitude) of $\tilde{S}_{\mathrm{sum}}$ for $k\in[K]$. In this case, we have $V=\Pi V(\mathcal{I},:)$.
\end{enumerate}
\end{lem}
Recall that $\Pi$ satisfies Equations (\ref{definePi}) and (\ref{pureNodes}), by $U=\Pi U(\mathcal{I},:)$  according to Lemma \ref{IS}, the rows of $U$ forms a $K$-simplex in $\mathbb{R}^{K}$, referred to as the Ideal Simplex of $U$ ($IS_{U}$ for brevity). Notably, the $K$ rows of $U(\mathcal{I},:)$ serve as its vertices. Given that $\Pi$ satisfies Equations (\ref{definePi}) and (\ref{pureNodes}), and $U=\Pi U(\mathcal{I},:)$, it follows that $U(i,:)=\Pi(i,:)U(\mathcal{I},:)$ for $i\in[n]$. This implies that $U(i,:)$ is a convex linear combination of the $K$ vertices of $IS_{U}$ with weights determined by $\Pi(i,:)$. Consequently, a pure row of $U$ (call $U(i,:)$ pure row if node $i$ is a pure node and mixed row otherwise) lies on one of the $K$ vertices of the simplex, while a mixed row occupies an interior position within the simplex. Such simplex structure is also observed in mixed membership estimation for community detection in single-layer networks \citep{mao2021estimating,jin2024mixed,qing2024bipartite} and in topic matrix estimation for topic modeling \citep{ke2024using}. Notably, the mixed membership matrix $\Pi$ can be precisely recovered by setting $\Pi=U U^{-1}(\mathcal{I},:)$ provided that the corner matrix $U(\mathcal{I},:)$ is obtainable. Thanks to the simplex structure in $U=\Pi U(\mathcal{I},:)$, the vertices of the simplex can be found by a vertex-hunting technique such as the successive projection (SP) algorithm \citep{araujo2001successive,gillis2013fast,gillis2015semidefinite}. Applying SP to all rows of $U$ with $K$ clusters enables the exact recovery of the $K\times K$ corner matrix $U(\mathcal{I},:)$. Details of the SP algorithm are provided in Algorithm 1 \citep{gillis2015semidefinite}. SP can efficiently find the $K$ distinct pure rows of $U$. Similarly, $V=\Pi V(\mathcal{I},:)$ also exhibits a simplex structure, leading to the recovery of $\Pi$ via $\Pi=V V^{-1}(\mathcal{I},:)$ by Lemma \ref{IS}. Consequently, applying the SP algorithm to $V$ with $K$ clusters also yields an exact reconstruction of the mixed membership matrix $\Pi$.
\begin{rem}
Let $\tilde{\mathcal{I}}$ represent another index set, consisting of arbitrary pure nodes $\tilde{p}_{1},\tilde{p}_{2},\ldots,\tilde{p}_{K}$ from respective communities $\mathcal{C}_{k}$ for $k\in[K]$. According to the proof of Lemma \ref{IS}, we have $U = \Pi U(\tilde{\mathcal{I}},:)$. Furthermore, this implies that $U = \Pi U(\tilde{\mathcal{I}},:) = \Pi U(\mathcal{I},:)$, which in turn signifies that $U(\tilde{\mathcal{I}},:)$ is equivalent to $U(\mathcal{I},:)$. In essence, the corner matrix $U(\mathcal{I},:)$ remains unchanged regardless of the specific pure nodes chosen to form the index set. An analogous argument holds for $V(\mathcal{I},:)$.
\end{rem}

For real-world multi-layer networks, the $L$ adjacency matrices $A_{1}, A_{2}, \ldots, A_{L}$ are observed instead of their expectations. For our first method, set $A_{\mathrm{sum}}=\sum_{l\in[L]}A_{l}$. We have $\mathbb{E}[A_{\mathrm{sum}}]=\Omega_{\mathrm{sum}}$ under MLMMSB. Subsequently, let $\hat{U}\hat{\Sigma}\hat{U}'$ be the top $K$ eigen-decomposition of $A_{\mathrm{sum}}$ where $\hat{U}$ is an orthogonal matrix with $\hat{U}'\hat{U}=I_{K\times K}$ and $\hat{\Sigma}$ is a diagonal matrix with its $k$-th diagonal element representing the $k$-th largest eigenvalue of $A_{\mathrm{sum}}$ in magnitude for $k\in[K]$. Given that the expectation of $A_{\mathrm{sum}}$ is $\Omega_{\mathrm{sum}}$, it follows that $\hat{U}$ can be interpreted as a slightly perturbed version of $U$. To proceed, we apply the SP algorithm to all rows of $\hat{U}$ with $K$ clusters, resulting in an estimated index set of pure nodes, denoted as $\hat{\mathcal{I}}$. We infer that  $\hat{U}(\hat{\mathcal{I}},:)$ should closely approximate the corner matrix $U(\mathcal{I},:)$. Finally, we estimate the mixed membership matrix $\Pi$ by computing $\hat{U}\hat{U}^{-1}(\hat{\mathcal{I}},:)$. Our first algorithm, called ``Successive projection on the sum of adjacency matrices'' (SPSum, Algorithm \ref{alg:SPSum}) summarises the above analysis.  It efficiently estimates the mixed membership matrix $\Pi$ by executing the SP algorithm on the top $K$ eigenvectors of  $A_{\mathrm{sum}}$.
\begin{algorithm}
\caption{SPSum}
\label{alg:SPSum}
\begin{algorithmic}[1]
\Require Adjacency matrices $A_{1}, A_{2}, \ldots, A_{L}$, and number of communities $K$.
\Ensure Estimated mixed membership matrix $\hat{\Pi}$.
\State Compute $A_{\mathrm{sum}}=\sum_{l\in[L]}A_{l}$.
\State Get $\hat{U}\hat{\Sigma}\hat{U}'$, the top $K$ eigen-decomposition of $A_{\mathrm{sum}}$.
\State Apply the SP algorithm on all rows of $\hat{U}$ with $K$ clusters to obtain the estimated index set $\hat{\mathcal{I}}$.
\State Set $\hat{\Pi}=\mathrm{max}(0,\hat{U}\hat{U}^{-1}(\hat{\mathcal{I}},:))$.
\State Normalize each row of $\hat{\Pi}$ by its $l_{1}$ norm: $\hat{\Pi}(i,:)=\frac{\hat{\Pi}(i,:)}{\|\hat{\Pi}(i,:)\|_{1}}$ for $i\in[n]$.
\end{algorithmic}
\end{algorithm}

In our second methodological approach, we define the aggregate matrix $S_{\mathrm{sum}}$ as the sum of the squared adjacency matrices, adjusted for bias, specifically as $S_{\mathrm{sum}}=\sum_{l\in[L]}(A_{l}^{2}-D_{l})$. Here, $D_{l}$ represents a diagonal matrix, with its $i$-th diagonal entry being the degree of node $i$ in layer $l$, i.e., $D_{l}(i,i)=\sum_{j\in[n]}A_{l}(i,j)$ for $i \in [n]$ and $l \in [L]$. The matrix $S_{\mathrm{sum}}$ was first introduced in \citep{lei2023bias} within the context of MLSBM, serving as a debiased estimate of the sum of squared adjacency matrices. This debiasing is necessary as $\sum_{l\in[L]} A_{l}^{2}$ alone  provides a biased approximation of $\sum_{l\in[L]}\Omega_{l}^{2}$. By subtracting the diagonal matrix $D_{l}$ from each squared adjacency matrix, we can effectively remove this bias, rendering $S_{\mathrm{sum}}$ a reliable estimator of $\tilde{S}_{\mathrm{sum}}$ as demonstrated in \citep{lei2023bias}. Subsequently, let $\hat{V}\hat{\Lambda}\hat{V}'$ be the top $K$ eigen-decomposition of $S_{\mathrm{sum}}$. Given that $\hat{V}$ provides a close approximation of $V$, applying the SP algorithm to all rows of $\hat{V}$ yields a reliable estimate of the corner matrix $V(\mathcal{I},:)$. In summary, our second spectral method, centered on $S_{\mathrm{sum}}$, is outlined in Algorithm \ref{alg:SPDSoS}. We refer to this method as ``Successive projection on the debiased sum of squared adjacency matrices" (SPDSoS for brevity).

The third method utilizes the summation of squared adjacency matrices, expressed as $\sum_{l\in[L]}A^{2}_{l}$, as a substitute for $S_{\mathrm{sum}}$ in Algorithm \ref{alg:SPDSoS}. Notably, this approach excludes the crucial bias-removal step inherent in $S_{\mathrm{sum}}$. We call this method ``Successive projection on the sum of squared adjacency matrices", abbreviated as SPSoS. In the subsequent section, we will demonstrate that SPDSoS consistently outperforms SPSoS, and both methods are generally theoretically superior to SPSum.
\begin{algorithm}
\caption{SPDSoS}
\label{alg:SPDSoS}
\begin{algorithmic}[1]
\Require Adjacency matrices $A_{1}, A_{2}, \ldots, A_{L}$, and number of communities $K$.
\Ensure Estimated mixed membership matrix $\hat{\Pi}$.
\State Compute $S_{\mathrm{sum}}=\sum_{l\in[L]}(A^{2}_{l}-D_{l})$.
\State Get $\hat{V}\hat{\Lambda}\hat{V}'$, the top $K$ eigen-decomposition of $S_{\mathrm{sum}}$.
\State Apple the SP algorithm on all rows of $\hat{V}$ with $K$ clusters to obtain the estimated index set $\hat{\mathcal{I}}$.
\State Set $\hat{\Pi}=\mathrm{max}(0,\hat{U}\hat{U}^{-1}(\hat{\mathcal{I}},:))$.
\State Update $\hat{\Pi}$ by $\hat{\Pi}(i,:)=\frac{\hat{\Pi}(i,:)}{\|\hat{\Pi}(i,:)\|_{1}}$ for $i\in[n]$.
\end{algorithmic}
\end{algorithm}
\section{Main results}\label{sec4}
In this section, we demonstrate the consistency of our methods by presenting their theoretical upper bounds for per-node error rates as the number of nodes $n$ and/or the number of layers $L$ increases within the context of the MLMMSB model. Assumption \ref{Assum1} provides a prerequisite lower bound for the sparsity parameter $\rho$ to ensure the theoretical validity of SPSum.
\begin{assum}\label{Assum1}
(sparsity requirement for SPSum) $\rho nL\geq\tau^{2}\mathrm{log}(n+L)$, where $\tau=\mathrm{max}_{i\in[n],j\in[n]}|\sum_{l\in[L]}(A_{l}(i,j)-\Omega_{l}(i,j))|$.
\end{assum}
In Assumption \ref{Assum1}, the choice of $\tau$ plays a crucial role in determining the level of sparsity regime. Since $\sum_{l\in[L]}A_{l}(i,j)\leq L$, $\tau$ has an upper bound $L$. However, to consider an even sparser scenario, we introduce a constant positive value $\beta$ that is strictly less than $L$ and remains fixed regardless of the growth of $n$ or $L$. This ensures that the aggregate number of edges connecting any two nodes $i$ and $j$ across all layers remains below $\beta$ even in the limit as $n$ or $L$ approaches infinity. Under these conditions, Assumption \ref{Assum1} is revised to reflect a more stringent requirement: $\rho\geq\frac{\beta^{2}\mathrm{log}(n+L)}{nL}$. This revised assumption characterizes the most challenging scenario for community detection. To establish the theoretical bounds for SPSum, we need the following requirement on $\mathcal{B}$.
\begin{assum}\label{Assum11}
(SPSum's requirement on $\mathcal{B}$) $|\lambda_{K}(\sum_{l\in[L]}B_{l})|\geq c_{1}L$ for some constant $c_{1}>0$.
\end{assum}
Assumption \ref{Assum11} is mild as $|\lambda_{K}(\sum_{l\in[L]}B_{l})|$ represents the smallest singular value resulting from the summation of $L$ matrices, and it is reasonable to presume that $|\lambda_{K}(\sum_{l\in[L]}B_{l})|$ is on the order of $O(L)$. To simplify our theoretical analysis, we also introduce the following condition:
\begin{con}\label{condition}
$K=O(1)$ and $\lambda_{K}(\Pi'\Pi)=O(\frac{n}{K})$.
\end{con}
Condition \ref{condition} is mild as $K=O(1)$ implies that the number of communities remains constant, while $\lambda_{K}(\Pi'\Pi)=O(\frac{n}{K})$ ensures the balance of community sizes, where the size of the $k$-th community $\mathcal{C}_{k}$ is defined as $\sum_{i\in[n]}\Pi(i,k)$ for $k\in[K]$. Our main result for SPSum offers an upper bound
on its per-node error rate in terms of $\rho, n$, and $L$ under the MLMMSB framework.
\begin{thm}\label{mainSPSum}
(Per-node error rate of SPSum) Under $\mathrm{MLMMSB}(\Pi,\rho,\mathcal{B})$, when Assumption \ref{Assum1}, Assumption \ref{Assum11}, and Condition \ref{condition} hold, let $\hat{\Pi}$ be obtained from the SPSum algorithm, there exists a $K\times K$ permutation matrix $\mathcal{P}$ such that with probability at least $1-o(\frac{1}{n+L})$, we have
\begin{align*}
\mathrm{max}_{i\in[n]}\|e'_{i}(\hat{\Pi}-\Pi\mathcal{P})\|_{1}=O(\sqrt{\frac{\mathrm{log}(n+L)}{\rho nL}}).
\end{align*}
\end{thm}
According to Theorem \ref{mainSPSum}, it becomes evident that as $n$ (and/or $L$) approaches infinity, SPSum's error rate diminishes to zero, highlighting the advantage of employing multiple layers in community detection and SPSum's consistent community detection. Additionally, Theorem \ref{mainSPSum} says that to attain a sufficiently low error rate for SPSum, the sparsity parameter $\rho$ must decrease at a rate slower than $\frac{\mathrm{log}(n+L)}{nL}$ (i.e., $\rho$ must be significantly greater than $\frac{\mathrm{log}(n+L)}{nL}$). This finding aligns with the sparsity requirement stated in Assumption \ref{Assum1} for SPSum. It is noteworthy that the theoretical upper bound for SPSum's per-node error rate, given as $O(\sqrt{\frac{\mathrm{log}(n+L)}{\rho nL}})$ in Theorem \ref{mainSPSum}, is independent of how we select the value of $\tau$ in Assumption \ref{Assum1}. This underscores the generality of our theoretical findings.

The assumption \ref{Assum2} stated below establishes the lower bound of the sparsity parameter $\rho$ that ensures the consistency of both SPDSoS and SPSoS.
\begin{assum}\label{Assum2}
(sparsity requirement for SPDSoS and SPSoS) $\rho^{2}n^{2}L\geq\tilde{\tau}^{2}\mathrm{log}(n+L)$, where $\tilde{\tau}=\mathrm{max}_{i\in[n]}\mathrm{max}_{j\in[n]}\\|\sum_{l\in[L]}\sum_{m\in[n]}(A_{l}(i,m)A_{l}(m,j)-\Omega_{l}(i,m)\Omega_{l}(m,j))|$.
\end{assum}
Analogous to Assumption \ref{Assum1}, when considering an extremely sparse scenario where $\tilde{\tau}$ does not exceed a positive constant $\tilde{\beta}$, Assumption \ref{Assum2} dictates that $\rho$ must satisfy $\rho\geq\frac{1}{n}\sqrt{\frac{\tilde{\beta}^{2}\mathrm{log}(n+L)}{L}}$. This requirement aligns with the sparsity requirement in Theorem 1 of \citep{lei2023bias}. The subsequent assumption serves a similar purpose to Assumption \ref{Assum11}.
\begin{assum}\label{Assum22}
(SPDSoS's and SPSoS's requirement on $\mathcal{B}$) $|\lambda_{K}(\sum_{l\in[L]}B^{2}_{l})|\geq c_{2}L$ for some constant $c_{2}>0$.
\end{assum}
Theorems \ref{mainSPDSoS} and \ref{mainSPSoS} are the main results of SPDSoS and SPSoS, respectively.
\begin{thm}\label{mainSPDSoS}
(Per-node error rate of SPDSoS) Under $\mathrm{MLMMSB}(\Pi,\rho,\mathcal{B})$, when Assumption \ref{Assum2}, Assumption \ref{Assum22}, and Condition \ref{condition} hold, let $\hat{\Pi}$ be obtained from the SPDSoS algorithm, there exists a $K\times K$ permutational matrix $\tilde{\mathcal{P}}$ such that with probability at least $1-o(\frac{1}{n+L})$, we have
\begin{align*}
\mathrm{max}_{i\in[n]}\|e'_{i}(\hat{\Pi}-\Pi\tilde{\mathcal{P}})\|_{1}=O(\sqrt{\frac{\mathrm{log}(n+L)}{\rho^{2}n^{2}L}})+O(\frac{1}{n}).
\end{align*}
\end{thm}
\begin{thm}\label{mainSPSoS}
(Per-node error rate of SPSoS) Under the same conditions as in Theorem \ref{mainSPDSoS} and suppose $\rho nL\geq\mathrm{log}(n+L)$, let $\hat{\Pi}$ be obtained from the SPSoS algorithm, with probability at least $1-o(\frac{1}{n+L})$, we have
\begin{align*}
\mathrm{max}_{i\in[n]}\|e'_{i}(\hat{\Pi}-\Pi\tilde{\mathcal{P}})\|_{1}=O(\sqrt{\frac{\mathrm{log}(n+L)}{\rho^{2}n^{2}L}})+O(\frac{1}{\rho n}).
\end{align*}
\end{thm}
According to Theorems \ref{mainSPDSoS} and \ref{mainSPSoS}, it is evident that both SPDSoS and SPSoS enjoy consistent mixed membership estimation. This consistency arises from the fact that their per-node error rates tend towards zero as the number of nodes $n$ (and/or the number of layers $L$) approaches infinity. Furthermore, similar to Theorem \ref{mainSPSum}, the theoretical bounds outlined in Theorems \ref{mainSPDSoS} and \ref{mainSPSoS} are independent of the selection of $\tilde{\tau}$ in Assumption \ref{Assum2}.

By the proof of Theorem \ref{mainSPSoS}, we know that SPSoS's theoretical upper bound of per-node error rate is the summation of SPDSoS's theoretical upper bound of per-node error rate and $O(\frac{1}{\rho n})$. Therefore, SPDSoS's error rate is always smaller than that of SPSoS and this indicates the benefit of the bias-removal step in $S_{\mathrm{sum}}$. Furthermore, to compare the theoretical performances between SPSum and SPSoS, we consider the following two simple cases:
\begin{itemize}
  \item When SPSoS's error rate is at the order $\frac{1}{\rho n}$, if SPSum significantly outperforms SPSoS, we have $\sqrt{\frac{\mathrm{log}(n+L)}{\rho nL}}\ll\frac{1}{\rho n}\Leftrightarrow L\gg\rho n\mathrm{log}(n+L)$.
  \item When SPSoS's error rate is at the order $\sqrt{\frac{\mathrm{log}(n+L)}{\rho^{2}n^{2}L}}$, if SPSum significantly outperforms SPSoS, we have $\sqrt{\frac{\mathrm{log}(n+L)}{\rho nL}}\ll\sqrt{\frac{\mathrm{log}(n+L)}{\rho^{2}n^{2}L}}\Leftrightarrow \rho n\ll1\Leftrightarrow \rho\ll\frac{1}{n}$. Since $\sqrt{\frac{\mathrm{log}(n+L)}{\rho nL}}\ll1$, we have $\rho\gg\frac{\mathrm{log}(n+L)}{nL}$. Combine $\rho\gg\frac{\mathrm{log}(n+L)}{nL}$ with $\rho\ll\frac{1}{n}$, we have $\frac{\mathrm{log}(n+L)}{nL}\ll\frac{1}{n}\Leftrightarrow L\gg\mathrm{log}(n+L)$.
\end{itemize}

The preceding points indicate that SPSum's superior performance over SPSoS is limited to scenarios where the number of layers $L$ is exceptionally large ($L\gg \rho n\mathrm{log}(n+L)$ or $L\gg\mathrm{log}(n+L)$). However, this requirement is unrealistic for the majority of real-world multi-layer networks. Given that SPDSoS consistently outperforms SPSoS, it follows that SPDSoS also prevails over SPSum in most scenarios. Additionally, even in the rare scenarios where SPSum significantly surpasses SPDSoS for a significantly large $L$, Theorem \ref{mainSPDSoS} assures that SPDSoS's error rate is negligible. Therefore, among these three methods, we arrive at the following conclusions: (a) SPDSoS consistently outperforms SPSoS; (b) SPSum's significant superiority over SPDSoS and SPSoS necessitates an unreasonably high number of layers, $L$, which is impractical for most real-world multi-layer networks. Consequently, we can confidently assert that SPDSoS and SPSoS virtually always surpass SPSum.
\section{Numerical results on synthetic multi-layer networks}\label{sec5}
In this section, we evaluate the performance of our proposed methods through the utilization of computer-generated multi-layer networks. For these simulated networks, we possess knowledge of the ground-truth mixed membership matrix $\Pi$. To quantify the performance of each method, we employ two metrics: the Hamming error and the Relative error. The Hamming error is defined as $\mathrm{Hamming ~error} = \mathrm{min}_{\mathcal{P}\in\mathcal{S}}\frac{\|\hat{\Pi}-\Pi\mathcal{P}\|_{1}}{n}$, while the relative error is given by $\mathrm{Relative ~error} = \mathrm{min}_{\mathcal{P}\in\mathcal{S}}\frac{\|\hat{\Pi}-\Pi\mathcal{P}\|_{F}}{\|\Pi\|_{F}}$. Here, $\mathcal{S}$ represents the set of all $K$-by-$K$ permutation matrices, accounting for potential label permutations. For each parameter setting in our simulation examples, we report the average of each metric for each proposed approach across 100 independent repetitions.

For all simulations conducted below, we set $K=3$ and let $n_{0}$ be the number of pure nodes within each community. For mixed nodes, say node $i$, we formulate its mixed membership vector $\Pi(i,:)$ in the following manner. Initially, we generate two random values $r^{(1)}_{i}=\mathrm{rand}(1)$ and $r^{(2)}_{i}=\mathrm{rand}(1)$, where $\mathrm{rand}(1)$ represents a random value drawn from a Uniform distribution over the interval $[0,1]$. Subsequently, we define the mixed membership vector as $\Pi(i,:)=(\frac{r^{(1)}_{i}}{2},\frac{r^{(2)}_{i}}{2},1-\frac{r^{(1)}_{i}}{2}-\frac{r^{(2)}_{i}}{2})$. Regarding the connectivity matrices, for the $l$-th matrix $B_{l}$, we assign $B_{l}(k,\tilde{k})=B_{l}(\tilde{k},k)=\mathrm{rand}(1)$ for all $1\leq k\leq \tilde{k}\leq K$ and $l\in[L]$. Finally, the number of nodes $n$, the sparsity parameter $\rho$, the number of layers $L$, and the number of pure nodes $n_{0}$ within each community are independently set for each simulation.

\emph{Experiment 1: Changing $\rho$.} Fix $(n,L,n_{0})=(500,100,100)$ and let $\rho$ range in$\{0.02,0.04,\ldots,0.2\}$. The multi-layer network becomes denser as $\rho$ increases. The results shown in panel (a) and panel (b) of Figure \ref{EX} demonstrate that SPDSoS performs slightly better than SPSoS, and both methods significantly outperform SPSum. Meanwhile, the error rates of SPDSoS and SPSoS decrease rapidly as the sparse parameter $\rho$ increases while SPSum's error rates decrease slowly.
\begin{figure}
\centering
\resizebox{\columnwidth}{!}{
\subfigure[Experiment 1: Hamming error]{\includegraphics[width=0.32\textwidth]{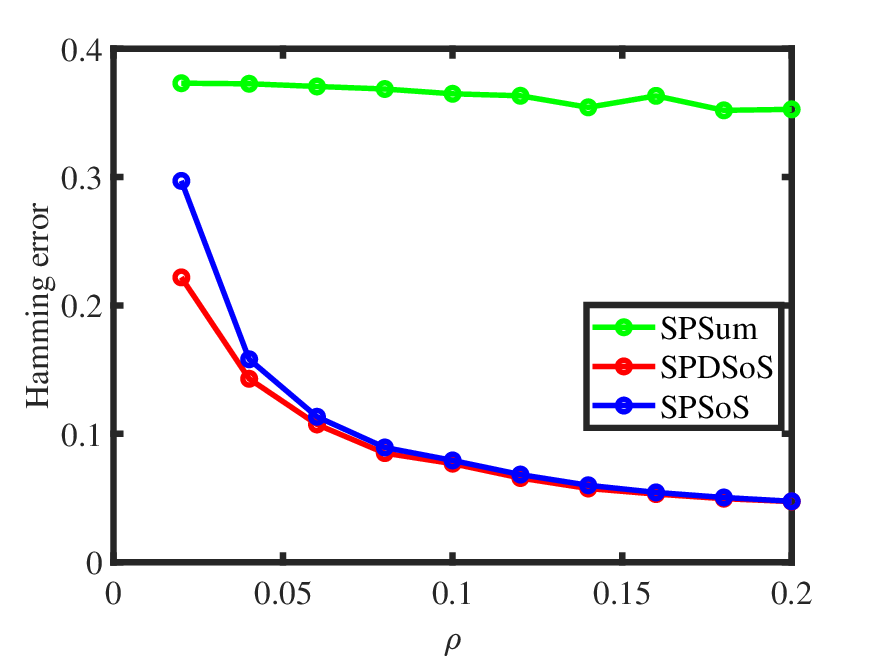}}
\subfigure[Experiment 1: Relative error]{\includegraphics[width=0.32\textwidth]{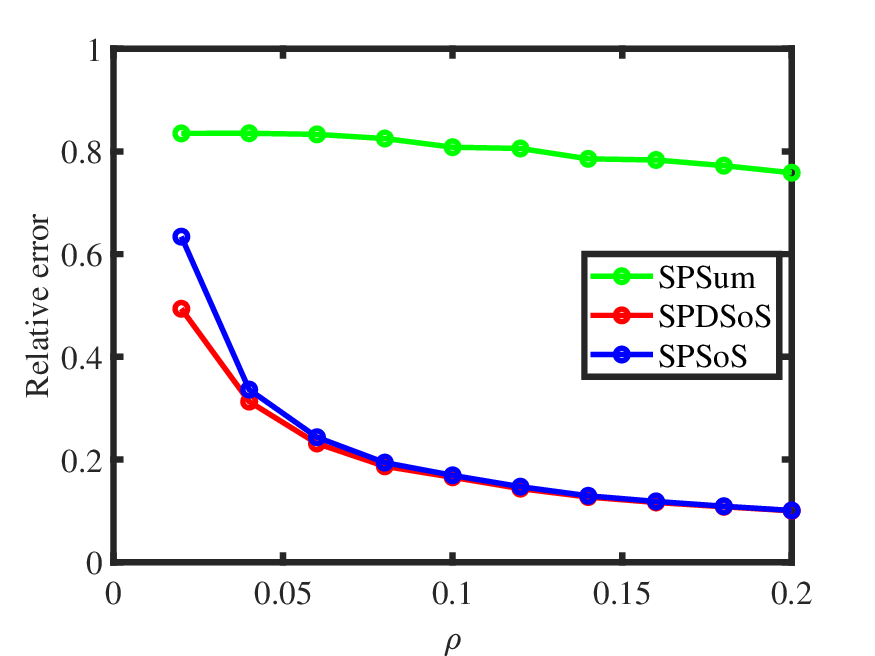}}
\subfigure[Experiment 2: Hamming error]{\includegraphics[width=0.32\textwidth]{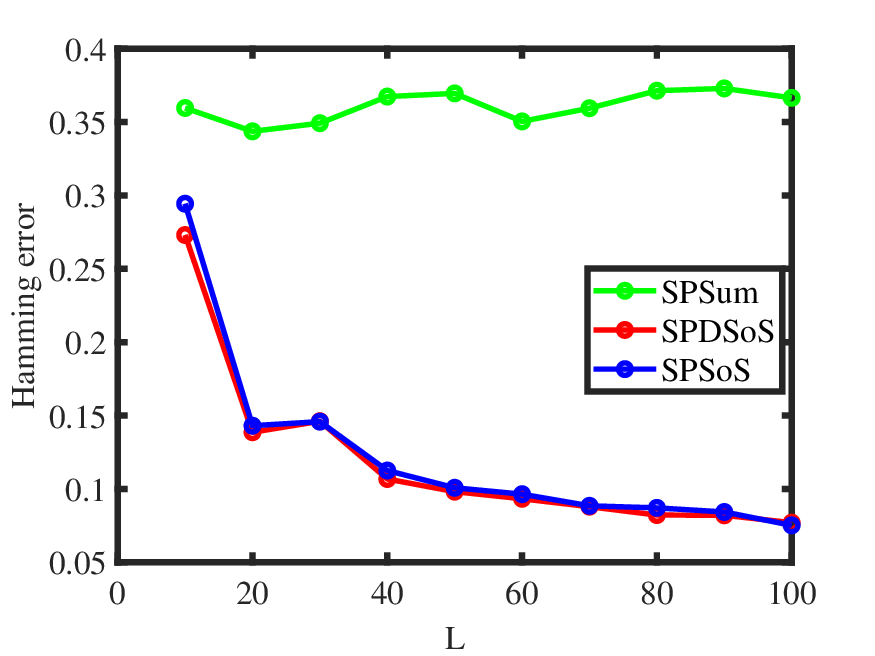}}
\subfigure[Experiment 2: Relative error]{\includegraphics[width=0.32\textwidth]{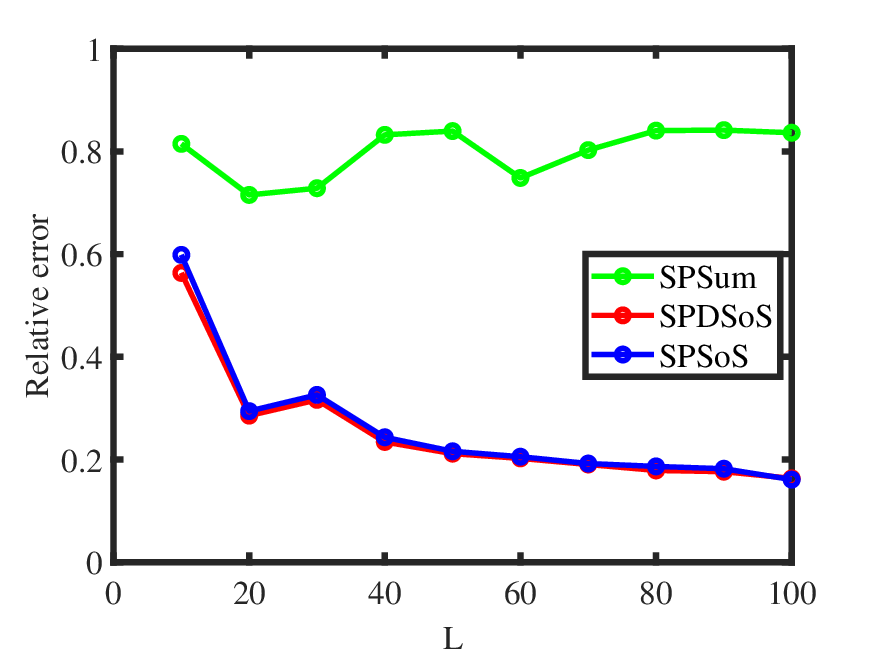}}
}
\resizebox{\columnwidth}{!}{
\subfigure[Experiment 3: Hamming error]{\includegraphics[width=0.32\textwidth]{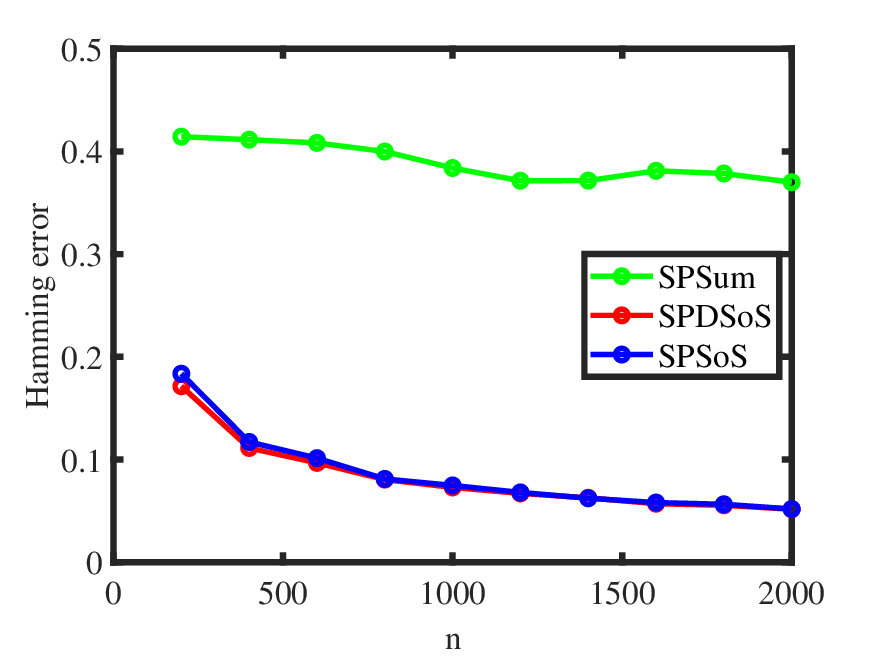}}
\subfigure[Experiment 3: Relative error]{\includegraphics[width=0.32\textwidth]{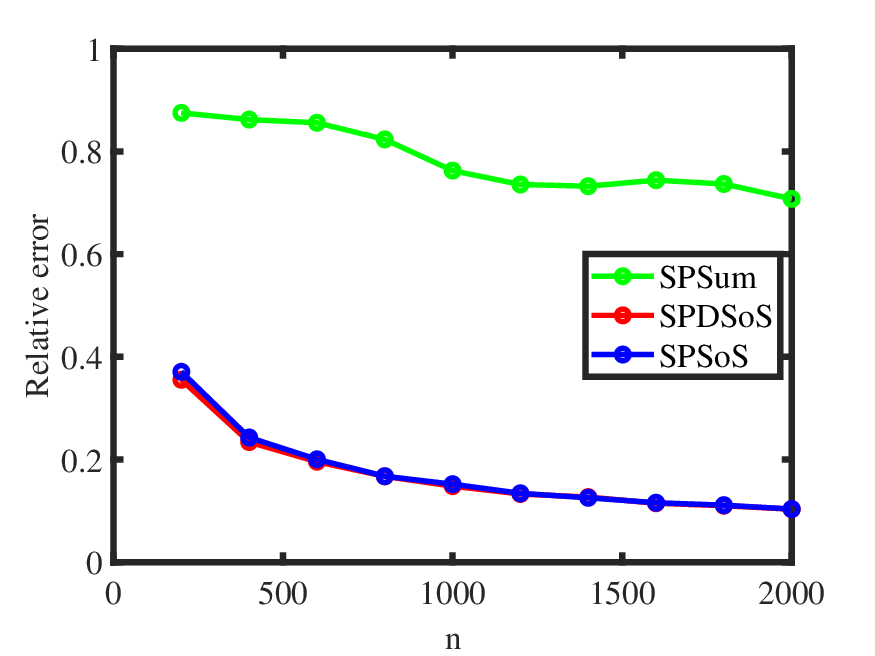}}
\subfigure[Experiment 4: Hamming error]{\includegraphics[width=0.32\textwidth]{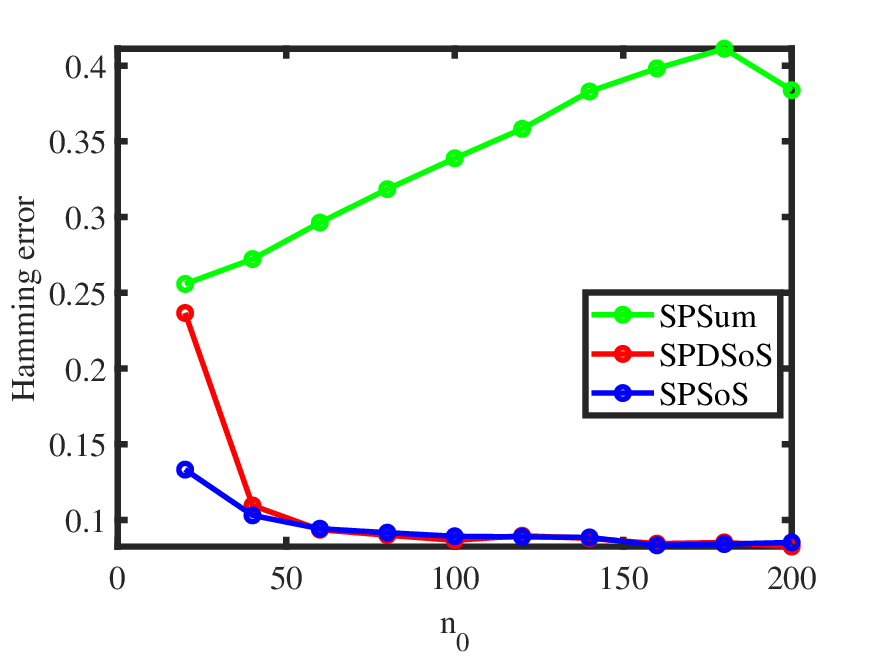}}
\subfigure[Experiment 4: Relative error]{\includegraphics[width=0.32\textwidth]{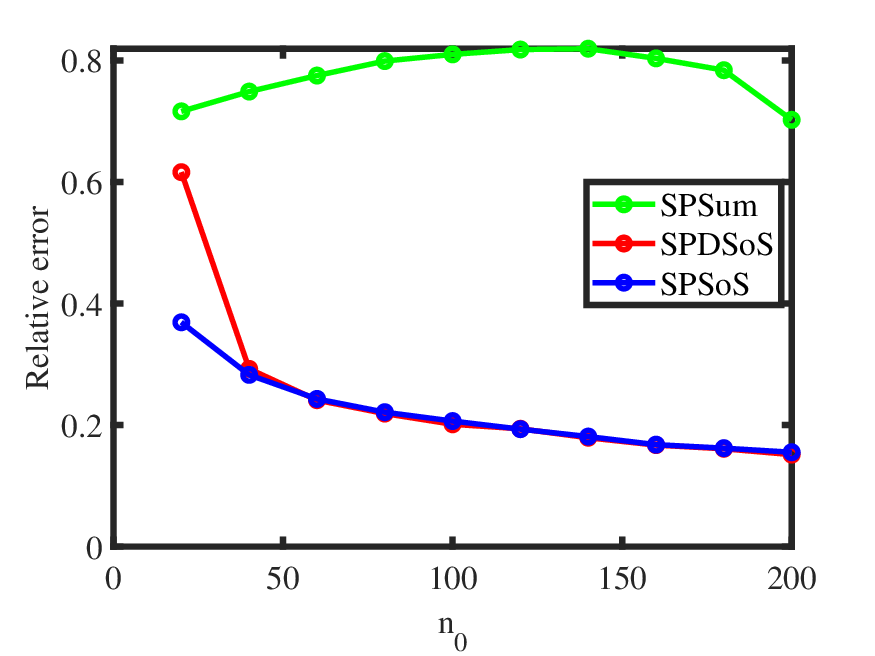}}
}
\caption{Numerical results.}
\label{EX} 
\end{figure}

\emph{Experiment 2: Changing $L$.} Fix $(n,\rho,n_{0})=(500,0.1,100)$ and let $L$ range in$\{10,20,\ldots,100\}$. More layers are observed as $L$ increases. The results are presented in panel (c) and panel (d) of Figure \ref{EX}. It is evident that SPDSoS and SPSoS exhibit comparable performance, both significantly surpassing SPSum. Furthermore, as $L$ increases, SPDSoS and SPSoS demonstrate improved performance, whereas SPSum's performance remains relatively unchanged throughout this experiment.

\emph{Experiment 3: Changing $n$.} Fix $(L,\rho)=(40,0.1)$, let $n$ range in $\{200,400,\ldots,2000\}$, and let $n_{0}=\frac{n}{4}$ for each choice of $n$. Panel (e) and panel (f) of Figure \ref{EX} display the results, which demonstrate that SPDSoS and SPSoS enjoy similar performances and that both methods perform better than SPSum. Meanwhile, we also observe that the error rates of all methods decrease as $n$ increases here.

\emph{Experiment 4: Changing $n_{0}$.} Fix $(n,L,\rho)=(600,50,0.1)$ and let $n_{0}$ range in $\{20,40,\ldots,200\}$. The number of pure nodes increases as $n_{0}$ grows. The results are displayed in the final two panels of Figure \ref{EX}. Our observations indicate that SPDSoS and SPSoS exhibit notably superior performance compared to SPSum. Furthermore, both SPDSoS and SPSoS exhibit improved performance in scenarios with a higher number of pure nodes, whereas SPSum demonstrates weaker performance in this experiment.
\section{Real data applications}\label{sec6realdata}
In this section, we demonstrate the application of our methods to multi-layer networks in the real world. The application of a mixed membership estimation algorithm to such networks consistently yields an estimated mixed membership matrix, denoted as $\hat{\Pi}$. Notably,  $\hat{\Pi}$ may vary depending on the specific algorithm used. Consequently, accurately assessing the quality of the estimated mixed membership community partition becomes a crucial problem. To address this challenge, we introduce two modularity metrics in this paper, designed to quantitatively evaluate the quality of mixed membership community detection in real-world multi-layer networks.

Recall that $A_{\text{sum}} = \sum_{l \in [L]} A_{l}$ represents the summation of all adjacency matrices. This summation effectively quantifies the weight between nodes. Given that nodes sharing similar community memberships tend to exhibit stronger connections than those with different memberships, $A_{\text{sum}}$ can be interpreted as the weighted adjacency matrix of an assortative network \citep{newman2002assortative,newman2003mixing}. Our fuzzy sum modularity, $Q_{\text{fsum}}$, is defined as follows:
\begin{align*}
Q_{\text{fsum}} = \frac{1}{m_{\text{sum}}} \sum_{i \in [n]} \sum_{j \in [n]} \left( A_{\text{sum}}(i,j) - \frac{d_{\text{sum}}(i)d_{\text{sum}}(j)}{m_{\text{sum}}} \right) \hat{\Pi}(i,:) \hat{\Pi}'(j,:),
\end{align*}
where $\text{fsum}$ stands for ``fuzzy sum'', $d_{\text{sum}}(i) = \sum_{j \in [n]} A_{\text{sum}}(i,j)$ for $i \in [n]$, and $m_{\text{sum}} = \sum_{i \in [n]} d_{\text{sum}}(i)$. Notably, when $L = 1$, our fuzzy sum modularity $Q_{\text{fsum}}$ simplifies to the fuzzy modularity introduced in Equation (14) of \citep{nepusz2008fuzzy}. Furthermore, when $L = 1$ and all nodes are pure, our modularity metric reduces to the classical Newman-Girvan modularity \citep{newman2004finding}.

Our second modularity metric is the fuzzy mean modularity, defined as the average of fuzzy modularities across all layers. Here's how it is formulated:
\begin{align*}
Q_{\text{fmean}} = \frac{1}{L} \sum_{l \in [L]} \sum_{i \in [n]} \sum_{j \in [n]} \frac{1}{m_{l}} \left( A_{l}(i,j) - \frac{D_{l}(i,i)D_{l}(j,j)}{m_{l}} \right) \hat{\Pi}(i,:) \hat{\Pi}'(j,:),
\end{align*}
where $\text{fmean}$ stands for ``fuzzy mean'' and $m_{l} = \sum_{i \in [n]} D_{l}(i,i)$ represents the sum of diagonal elements in the degree matrix $D_{l}$ for layer $l$ for $l\in[L]$. When all nodes are pure, our fuzzy mean modularity $Q_{\text{fmean}}$ simplifies to the multi-normalized average modularity introduced in Equation (3.1) of \cite{paul2021null}. When there is only a single layer ($L=1$), $Q_{\text{fmean}}$ reduces to the fuzzy modularity described in \citep{nepusz2008fuzzy}. Furthermore, if both conditions hold ($L=1$ and all nodes are pure), $Q_{\text{fmean}}$ degenerates to the classical Newman-Girvan modularity \citep{newman2004finding}.

Analogous to the Newman-Girvan modularity \citep{newman2004finding}, the fuzzy modularity \citep{nepusz2008fuzzy}, and the multi-normalized average modularity \cite{paul2021null}, a higher value of the fuzzy sum modularity $Q_{\text{fsum}}$ indicates a better community partition. Consequently, we consistently favor a larger value of $Q_{\text{fsum}}$. Similarly, a larger $Q_{\text{fmean}}$ also indicates a better community partition. To the best of our knowledge, our fuzzy sum modularity $Q_{\text{fsum}}$ and fuzzy mean modularity $Q_{\text{fmean}}$ are the first metrics designed to evaluate the quality of mixed membership community detection in multi-layer networks.

For real-world multi-layer networks, where the number of communities $K$ is usually unknown, we adopt the strategy introduced in \citep{nepusz2008fuzzy} to estimate $K$. Specifically, we determine $K$ by selecting the one that maximizes the fuzzy sum modularity $Q_{\text{fsum}}$ (or the fuzzy mean modularity $Q_{\text{fmean}}$). This strategy ensures that we obtain an optimal community partition based on the chosen modularity metric.

In this paper, we consider the following real-world multi-layer networks, which can be accessed at \url{https://manliodedomenico.com/data.php}:
\begin{itemize}
\item \textbf{Lazega Law Firm}: This data is a multi-layer social network with $n=71$ nodes and $L=3$ layers \citep{lazega2001collegial,snijders2006new}. For this data, nodes represent company partners and layers denote different cooperative relationships (Advice, Friendship, and Co-work) among partners.
  \item \textbf{C.Elegans}: This data is from biology and it collects the connection of Caenorhabditis elegans  \citep{chen2006wiring}. It has $n=279$ nodes and $L=3$ layers with nodes denoting Caenorhabditis elegans and layers being different synaptic junctions (Electric, Chemical Monadic, and Chemical Polyadic).
  \item \textbf{CS-Aarhus}: This data is a multi-layer social network with $n=61$ nodes and $L=5$ layers, where nodes represent employees of the Computer Science department at Aarhus and layers denote different relationships (Lunch, Facebook, Coauthor, Leisure, Work) \citep{magnani2013combinatorial}.
  \item \textbf{FAO-trade}: This data collects different types of trade relationships among countries from FAO (Food and Agriculture Organization of the United Nations) \citep{de2015structural}. It has $n=214$ nodes and $L=364$ layers, where nodes represent countries, layers denote products, and edges denote trade relationships among countries.
\end{itemize}

The estimated number of communities and the corresponding modularity of each method for the real data analyzed in this paper are comprehensively presented in Table \ref{realdataQfsum} and Table \ref{realdataQfmean}. After analyzing the results in these tables, we arrive at the following conclusions:
\begin{itemize}
\item $Q_{\text{fsum}}$ and $Q_{\text{fmean}}$ consistently demonstrate similar values for each method, indicating a high degree of agreement. Specifically, SPSum achieves the highest modularity score using both metrics across multiple datasets, including Lazega Law Firm, CS-Aarhus, and FAO-trade. It is worth mentioning that this observation (SPSum surpassing the other two methods) contrasts with both theoretical and numerical discoveries. A plausible explanation for this discrepancy could be that $Q_{\text{fsum}}$ and $Q_{\text{fmean}}$ are computed directly from $\{A_{l}\}^{L}_{l=1}$ rather than $\{A^{2}_{l}\}^{L}_{l=1}$. Similarly, for SPSoS, its $Q_{\text{fsum}}$ and $Q_{\text{fmean}}$ scores for CS-Aarhus rank second among all methods, with scores that are closely aligned. This consistency across measures suggests that both $Q_{\text{fsum}}$ and $Q_{\text{fmean}}$ effectively capture similar aspects of community structure, providing reliable and comparable assessments of different methods.
  \item While SPSum may not perform as well in numerical studies, it nevertheless exhibits superior performance in terms of modularity for real-world datasets compared to SPDSoS and SPSoS. The only exception is the C.Elegans network, where SPSum's modularity score is slightly lower than SPDSoS's.
  \item The CS-Aarhus network exhibits a more distinct community structure compared to the other three real multi-layer networks, as evidenced by the higher modularity scores obtained by our methods for CS-Aarhus. Conversely, the FAO-trade network possesses the least discernible community structure among all real multi-layer networks, as the modularity scores achieved by all proposed methods for FAO-trade are lower than those of the other three networks.
  \item The results presented in Table \ref{realdataQfsum} and Table \ref{realdataQfmean} indicate that the optimal number of communities for the four real-world multi-layer networks, namely Lazega Law Firm, CS-Aarhus, CS-Aarhus, and FAO-trade, are 3, 2, 5, and 2, respectively. This determination is based on selecting the value of $K$ that yields the highest modularity score across all methods for each dataset.
\end{itemize}
\begin{table}[h!]
\centering
\caption{$(\mathrm{Estimated~}K, Q_{\text{fsum}})$ obtained by the proposed approaches for the real data used in this paper. The boldface values represent the highest $Q_{\text{fsum}}$ scores among the three methods.}
\label{realdataQfsum}
\begin{tabular}{cccccccccccc}
\hline\hline
Dataset&SPSum&SPDSoS&SPSoS\\
\hline
Lazega Law Firm&(3,\textbf{0.2025})&(3,0.1604)&(3,0.1993)\\
C.Elegans&(2,0.2778)&(2,0.\textbf{2808})&(2,0.2509)\\
CS-Aarhus&(5,\textbf{0.3575})&(4,0.3474)&(4,0.3559)\\
FAO-trade&(2,\textbf{0.1508})&(2,0.1329)&(2,0.1412)\\
\hline\hline
\end{tabular}
\end{table}

\begin{table}[h!]
\centering
\caption{$(\mathrm{Estimated~}K, Q_{\text{fmean}})$ obtained by the proposed approaches for the real data used in this paper. The boldface values represent the highest $Q_{\text{fmean}}$ scores among the three methods.}
\label{realdataQfmean}
\begin{tabular}{cccccccccccc}
\hline\hline
Dataset&SPSum&SPDSoS&SPSoS\\
\hline
Lazega Law Firm&(3,\textbf{0.1990})&(3,0.1572)&(3,0.1961)\\
C.Elegans&(2,0.2779)&(2,\textbf{0.2780})&(2,0.2495)\\
CS-Aarhus&(5,\textbf{0.3681})&(4,0.3404)&(4,0.3529)\\
FAO-trade&(2,\textbf{0.1487})&(2,0.1274)&(2,0.1403)\\
\hline\hline
\end{tabular}
\end{table}

To simplify our analysis and consider the fact that SPSum generates modularity scores that are either higher or comparable to those of SPDSoS and SPSoS for the four real-world multi-layer networks under consideration, we focus our further analysis on the SPSum method. Let $\hat{\Pi}$ represent the estimated mixed membership matrix obtained by SPSum for a given real-world multi-layer network. For a node $i$ within the network, we define its estimated home base community as the one that corresponds to the maximum value in the $i$-th row of $\hat{\Pi}$, denoted as $\tilde{k}=\mathrm{argmax}_{k\in[K]}\hat{\Pi}(i,k)$. Furthermore, we categorize nodes based on their membership distribution: a node is considered highly mixed if $\mathrm{max}_{k\in[K]}\hat{\Pi}(i,k)\leq0.6$, highly pure if $\mathrm{max}_{k\in[K]}\hat{\Pi}(i,k)\geq0.9$, and neutral otherwise. We introduce two additional metrics, $\varsigma_{mixed}$ and $\varsigma_{pure}$, which represent the proportions of highly mixed and highly pure nodes, respectively. Additionally, we define the balanced parameter $\upsilon$ as the ratio of the minimum to the maximum $l_1$ norm of the columns of $\hat{\Pi}$, i.e., $\upsilon=\frac{\mathrm{min}_{k\in[K]}\|\hat{\Pi}(:,k)\|_{1}}{\mathrm{max}_{k\in[K]}\|\hat{\Pi}(:,k)\|_{1}}$. A higher value of $\upsilon$ indicates a more balanced multi-layer network. Table \ref{realdata3indices} presents the values of these three indices obtained by applying SPSum to the real-world multi-layer networks studied in this paper. Based on the results in Table \ref{realdata3indices}, we draw the following conclusions:
\begin{itemize}
  \item Across all networks, the number of highly pure nodes significantly exceeds the number of highly mixed nodes, indicating that most nodes are strongly associated with a single community, while only a few exhibit mixed membership.
  \item In the Lazega Law Firm network, approximately 10 nodes are highly mixed, 35 nodes are highly pure, and 26 are neutral. Meanwhile, this network exhibits the highest proportion of highly mixed nodes among the four real-world multi-layer networks analyzed.
  \item The C.Elegans network has the lowest proportion of highly mixed nodes and the highest proportion of highly pure nodes. Its balanced parameter, $\upsilon$, is 0.9512, the highest among all networks, indicating that the sizes of the two estimated communities in C.Elegans are nearly identical.
  \item The CS-Aarhus network exhibits indices that are comparable to those of the Lazega Law Firm network.
  \item FAO-trade displays a proportion of highly mixed (and pure) nodes similar to C.Elegans. However, its balanced parameter is the lowest among all networks, indicating that FAO-trade exhibits the most unbalanced community structure among the four real-world datasets.
\end{itemize}

\begin{table}[h!]
\footnotesize
	\centering
	\caption{$\varsigma_{mixed},\varsigma_{pure}$, and $\upsilon$ computed from $\hat{\Pi}$, where $\hat{\Pi}$ is returned by running the SPSum algorithm to real multi-layer networks used in this paper. Here, we use the Estimated $K$ of SPSum in Table \ref{realdataQfsum} for each data.}
	\label{realdata3indices}
	\begin{tabular}{cccccccccccc}
\hline\hline
Dataset&$\varsigma_{\mathrm{mixed}}$&$\varsigma_{\mathrm{pure}}$&$\upsilon$\\
\hline
Lazega Law Firm&0.1408&0.4930&0.7276\\
C.Elegans&0.0609&0.6882&0.9512\\
CS-Aarhus&0.1311&0.4590&0.7006\\
FAO-trade&0.0701&0.6449&0.5480\\
\hline\hline
\end{tabular}
\end{table}

\begin{figure}
\centering
\includegraphics[width=0.4\textwidth]{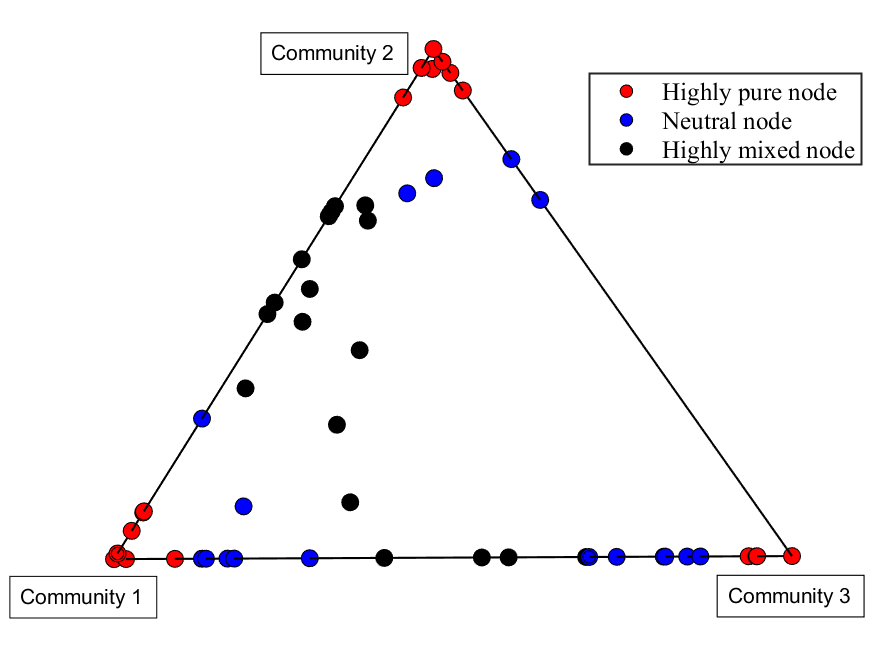}
\caption{Ternary diagram of the $71\times3$ estimated membership matrix $\hat{\Pi}$ for Lazega Law Firm. Each dot represents a company partner, and its location within the triangle corresponds to its membership scores.}
\label{LLFPi} 
\end{figure}

\begin{figure}
\centering
\resizebox{\columnwidth}{!}{
\subfigure[]{\includegraphics[width=0.2\textwidth]{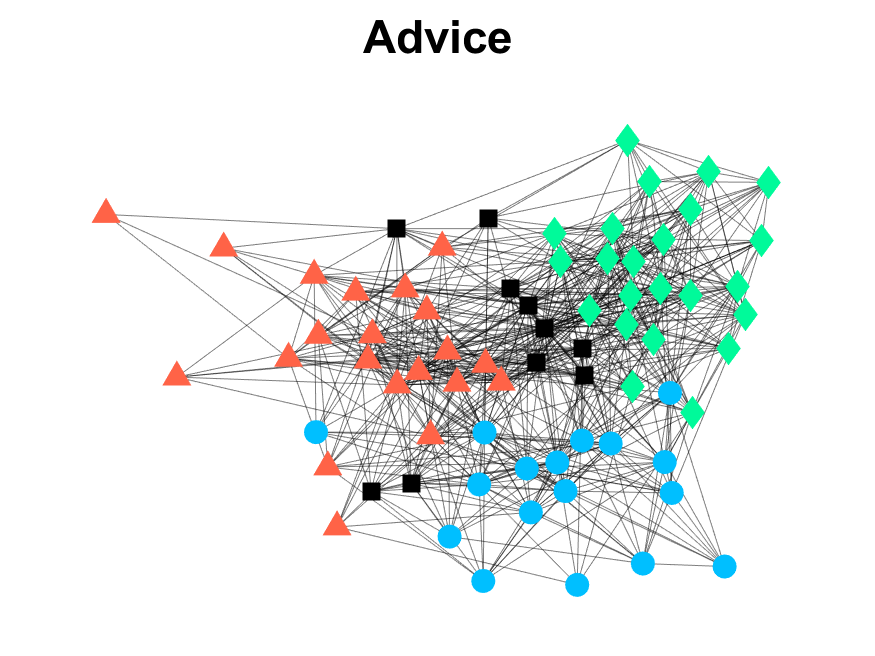}}
\subfigure[]{\includegraphics[width=0.2\textwidth]{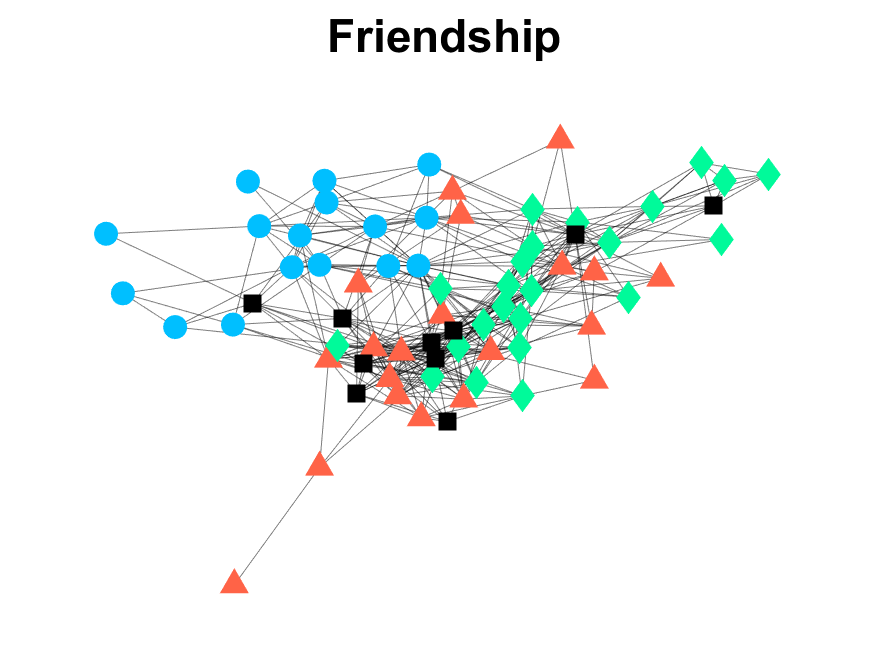}}
\subfigure[]{\includegraphics[width=0.2\textwidth]{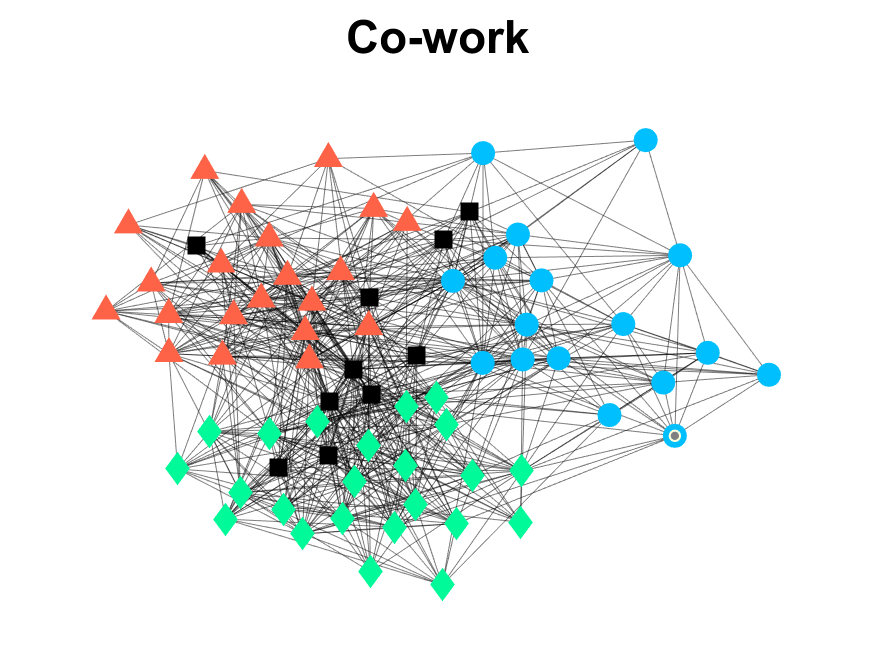}}
}
\caption{Illustration of the 3 layers of Lazega Law Firm. Colors (shapes) indicate home based communities and the black square represents highly mixed nodes detected by SPSum with 3 communities. For all layers, we only plot the largest connected graph.}
\label{LLFN} 
\end{figure}

\begin{figure}
\centering
\resizebox{\columnwidth}{!}{
\subfigure[]{\includegraphics[width=0.2\textwidth]{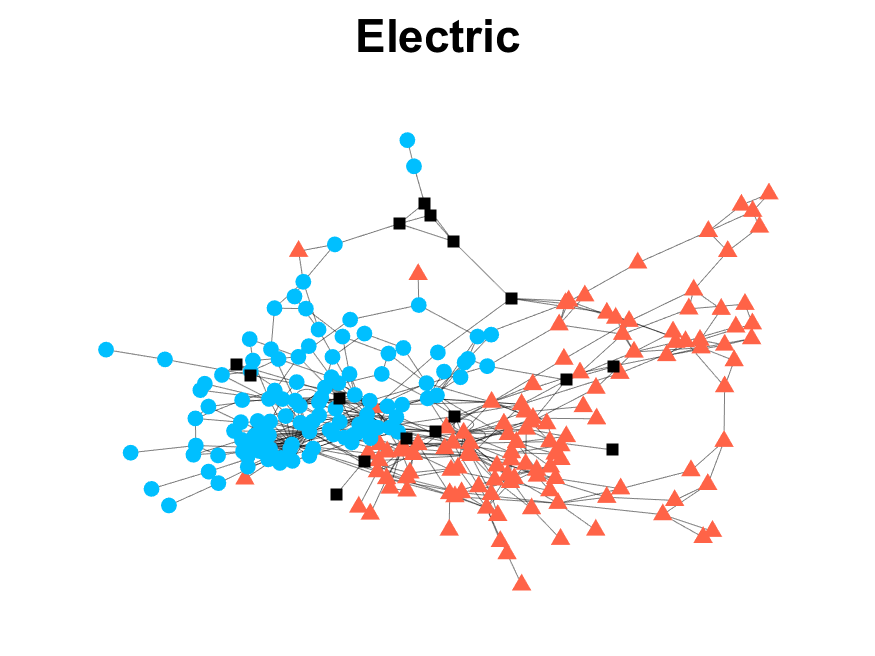}}
\subfigure[]{\includegraphics[width=0.2\textwidth]{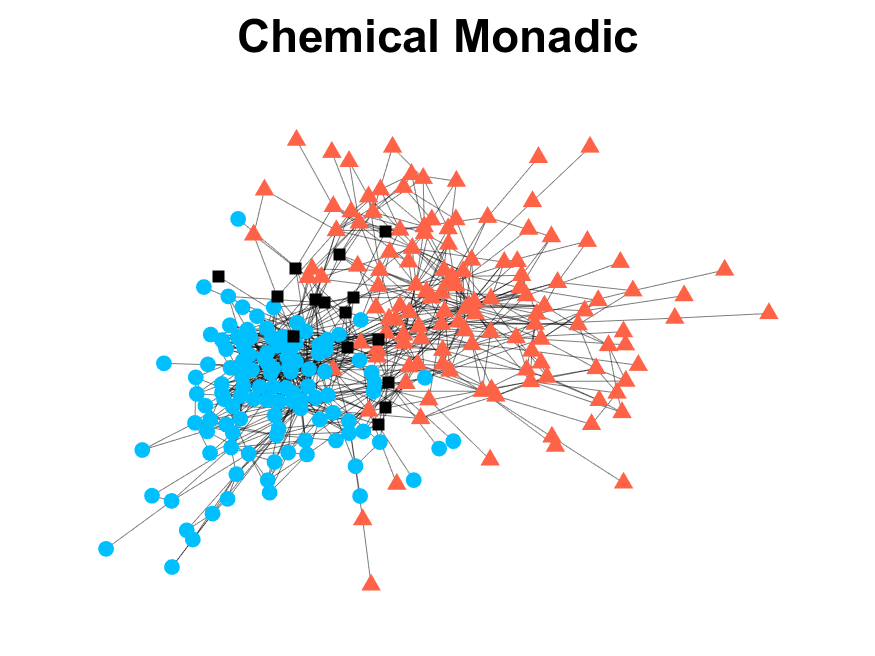}}
\subfigure[]{\includegraphics[width=0.2\textwidth]{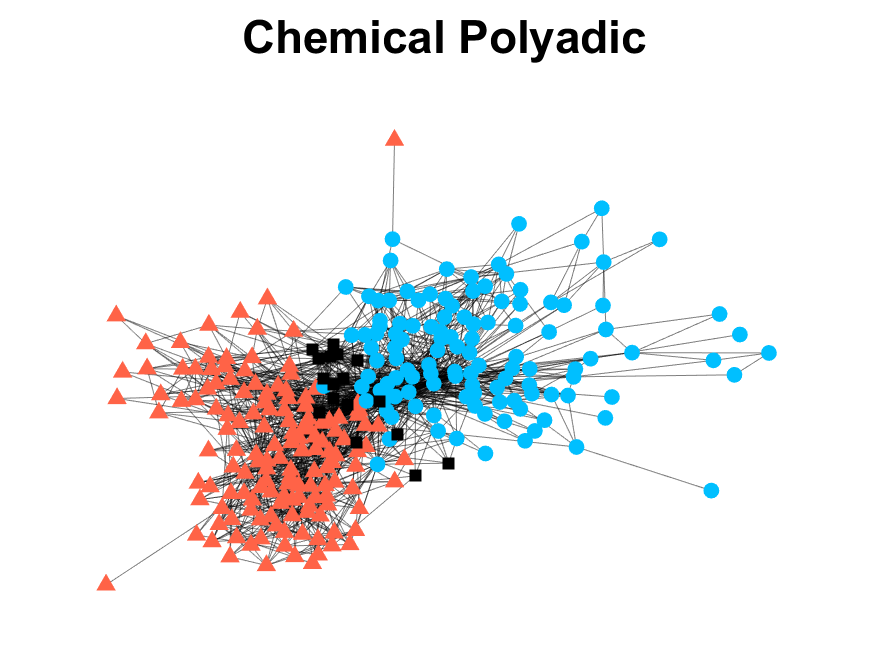}}
}
\caption{Illustration of the 3 layers of C.Elegans. Colors (shapes) indicate home based communities and the black square represents highly mixed nodes detected by SPSum with 2 communities. For all layers, we only plot the largest connected graph.}
\label{CEN} 
\end{figure}

\begin{figure}
\centering
\resizebox{\columnwidth}{!}{
\subfigure[]{\includegraphics[width=0.2\textwidth]{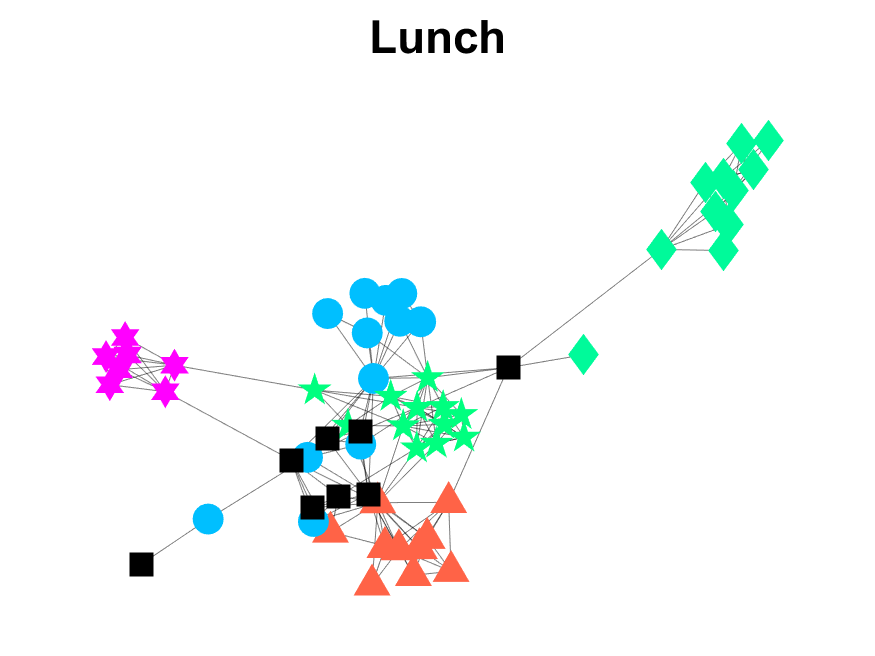}}
\subfigure[]{\includegraphics[width=0.2\textwidth]{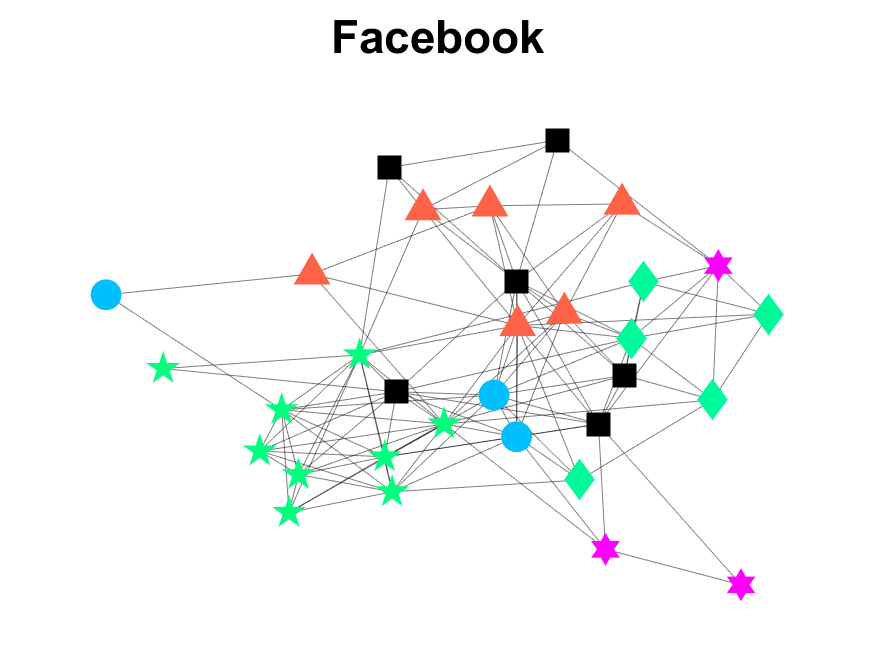}}
\subfigure[]{\includegraphics[width=0.2\textwidth]{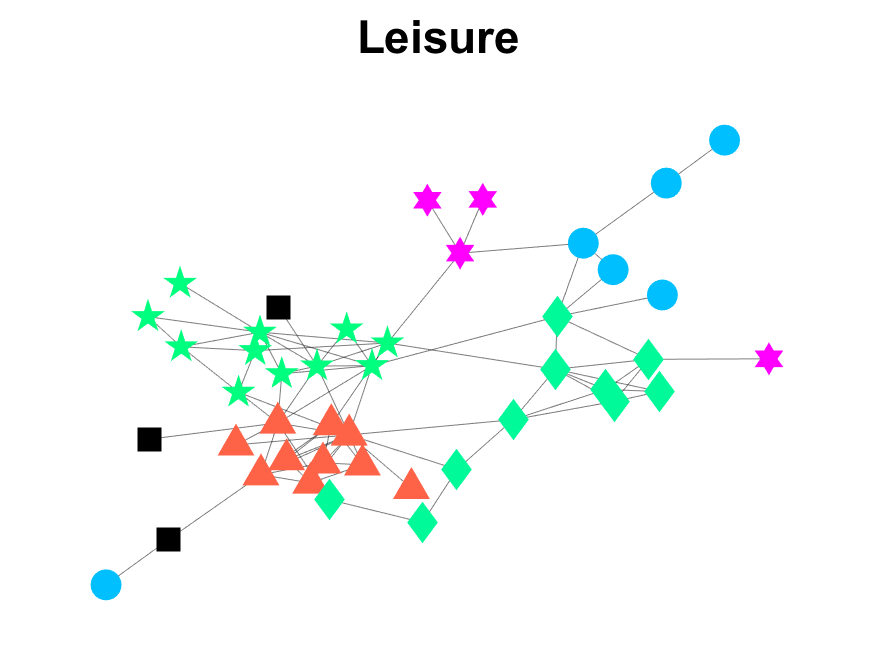}}
\subfigure[]{\includegraphics[width=0.2\textwidth]{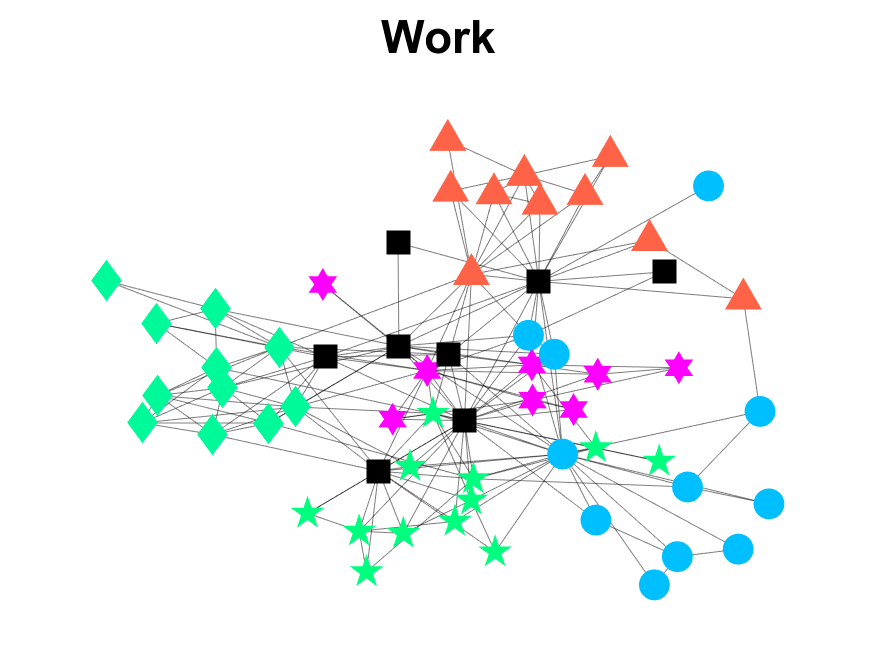}}
}
\caption{Illustration of the Lunch, Facebook, Leisure, and Work layers of CS-Aarhus, where we do not show the Coauthor layer because it is too sparse. Colors (shapes) indicate home based communities and the black square represents highly mixed nodes detected by SPSum with 5 communities. For all layers, we only plot the largest connected graph.}
\label{CSN} 
\end{figure}

\begin{figure}
\centering
\resizebox{\columnwidth}{!}{
\subfigure[]{\includegraphics[width=0.2\textwidth]{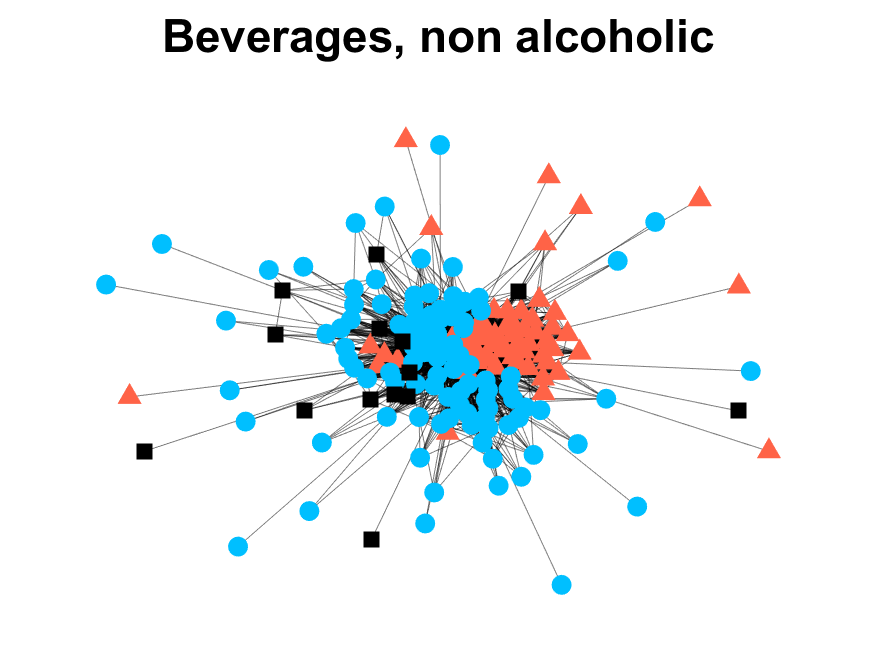}}
\subfigure[]{\includegraphics[width=0.2\textwidth]{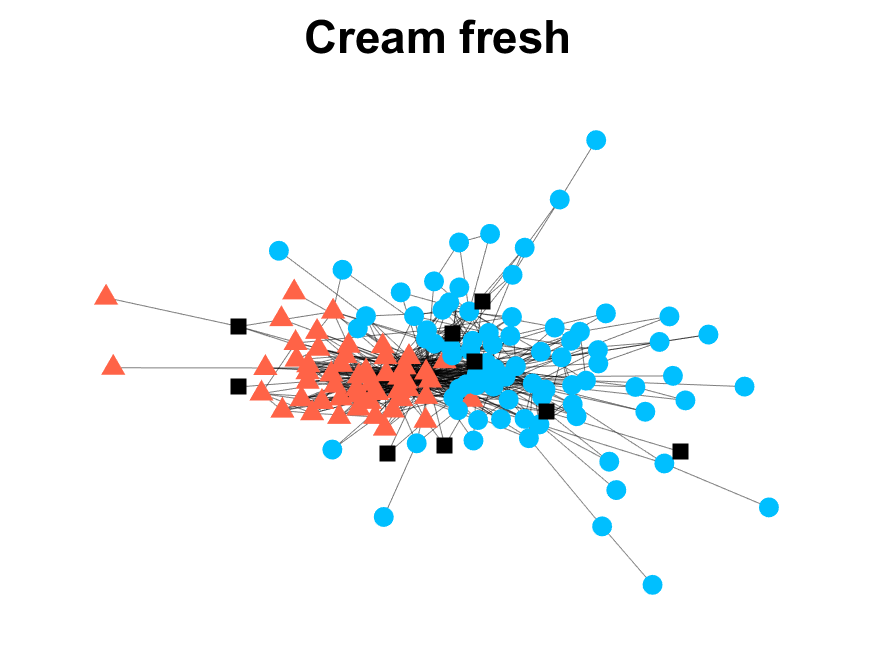}}
\subfigure[]{\includegraphics[width=0.2\textwidth]{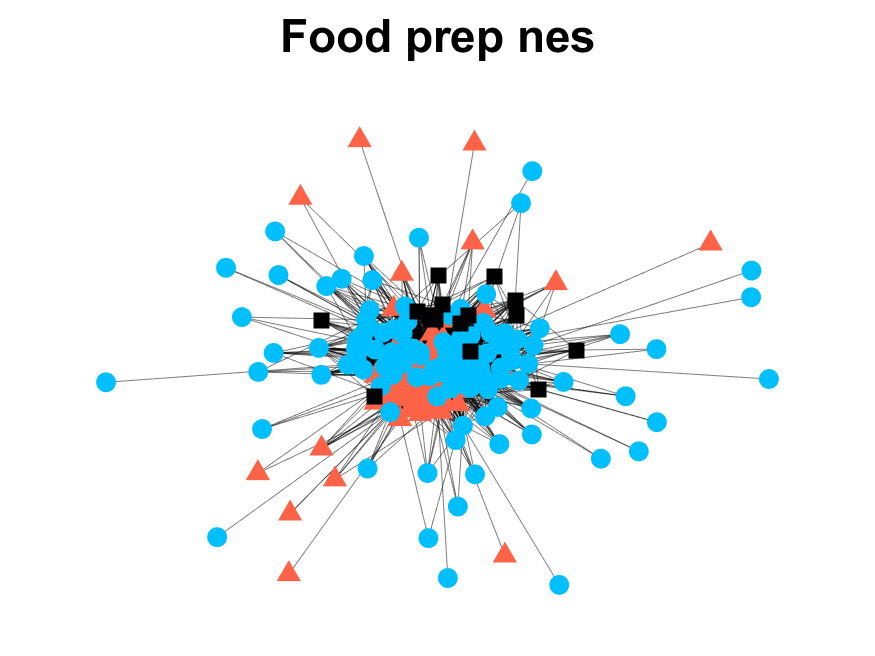}}
}
\resizebox{\columnwidth}{!}{
\subfigure[]{\includegraphics[width=0.2\textwidth]{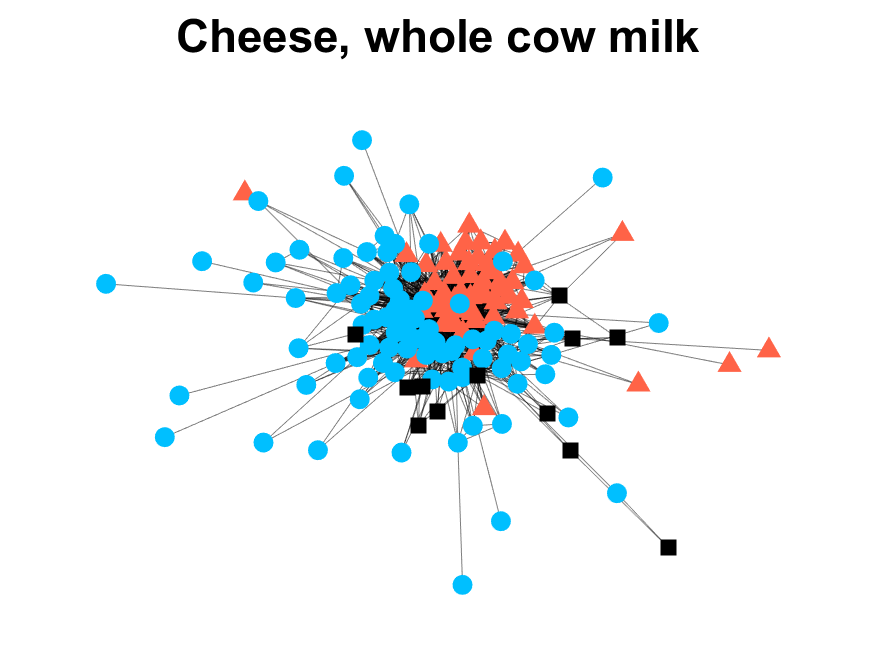}}
\subfigure[]{\includegraphics[width=0.2\textwidth]{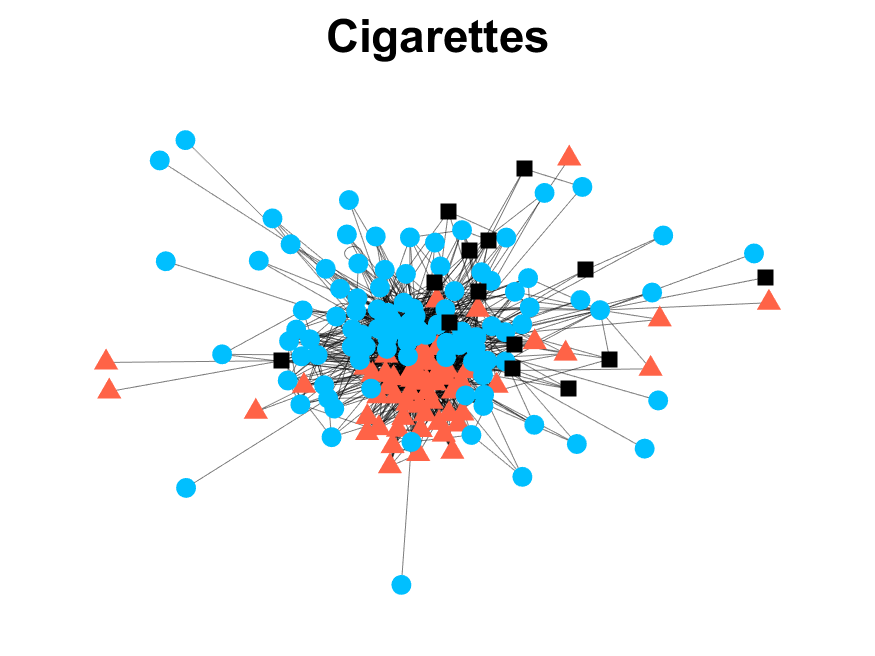}}
\subfigure[]{\includegraphics[width=0.2\textwidth]{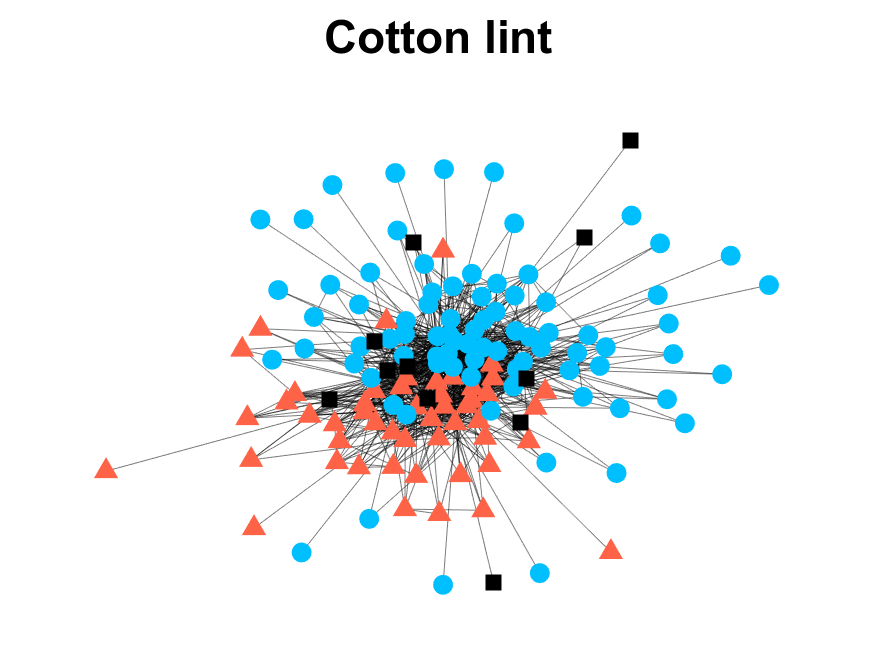}}
}
\caption{Illustration of the six layers of six products in FAO-trade. Colors (shapes) indicate home based communities and the black square represents highly mixed nodes detected by SPSum with 2 communities. For all layers, we only plot the largest connected graph.}
\label{FAON} 
\end{figure}

Figure \ref{LLFPi} presents a ternary diagram that visualizes the estimated community membership matrix $\hat{\Pi}$ obtained from the SPSum algorithm for the Lazega Law Firm network when there are three communities. In this diagram, we observe that node $i$ is positioned closer to one of the triangle's vertices compared to node $j$ if $\mathrm{max}_{k\in[3]}\hat{\Pi}(i,k)>\mathrm{max}_{k\in[3]}\hat{\Pi}(j,k)$ for $i\in[71], j\in[71]$. A pure node is located at one of the triangle's vertices, and a neutral node is closer to a vertex than a highly mixed node. Therefore, Figure \ref{LLFPi} indicates the purity of each node for the Lazega Law Firm network.

Figures \ref{LLFN}-\ref{FAON} present the communities estimated by our SPSum method for the four real multi-layer networks. In these figures, nodes sharing the same color represent nodes belonging to the same home base community, while black squares denote highly mixed nodes. From these figures, we can clearly find the communities detected by SPSum in each layer for each multi-layer network.
\section{Conclusion}\label{sec7}
This paper considers the problem of estimating community memberships of nodes in multi-layer networks under the multi-layer mixed membership stochastic block model, a model that permits nodes to belong to multiple communities simultaneously. We have developed spectral methods leveraging the eigen-decomposition of some aggregate matrices, and have provided theoretical guarantees for their consistency, demonstrating the convergence of per-node error rates as the number of nodes and/or layers increases under MLMMSB. To the best of our knowledge, this is the first work to estimate the mixed community memberships for multi-layer networks by using these aggregate matrices. Our theoretical analysis reveals that the algorithm designed based on the debiased sum of squared adjacency matrices always outperforms the algorithm using the sum of squared adjacency matrices, while they generally outperform the method using the sum of adjacency matrices in the task of estimating mixed memberships for multi-layer networks. Such a result is new for mixed membership estimation in multi-layer networks to the best of our knowledge. Extensive simulated studies support our theoretical findings, validating the efficiency of our method using the debiased sum of adjacency matrices. Additionally, the proposed fuzzy modularity measures offer a novel perspective for evaluating the quality of mixed membership community detection in multi-layer networks.

For future research, first, developing methods with theoretical guarantees for estimating the number of communities in MLMMSB remains a challenging and meaningful task. Second, accelerating our methods for detecting mixed memberships in large-scale multi-layer networks is crucial for practical applications. Third, exploring more efficient algorithms for estimating mixed memberships would further enrich our understanding of community structures in multi-layer networks. Finally, extending our framework to directed multi-layer networks would broaden the scope of our work and enable the analysis of even more complex systems.
\section*{CRediT authorship contribution statement}
\textbf{Huan Qing:} Conceptualization; Data curation; Formal analysis; Funding acquisition; Methodology; Project administration; Resources; Software; Validation; Visualization; Writing-original draft; Writing-review $\&$ editing.
\section*{Declaration of competing interest}
The author declares no competing interests.
\section*{Data availability}
Data and code will be made available on request.
\section*{Acknowledgements}
H.Q. was sponsored by the Scientific Research Foundation of Chongqing University of Technology (Grant No: 0102240003) and the Natural Science Foundation of Chongqing, China (Grant No: CSTB2023NSCQ-LZX0048).
\appendix
\section{Proofs}\label{SecProofs}
\subsection{Proof of Lemma \ref{IS}}
\begin{proof}
Since $\Omega_{\mathrm{sum}}=\rho\Pi(\sum_{l\in[L]}B_{l})\Pi'=U\Sigma U'$ and $U'U=I_{K\times K}$, we have
\begin{align*}
\rho\Pi(\sum_{l\in[L]}B_{l})\Pi'=U\Sigma U'\Rightarrow \Pi(\rho\sum_{l\in[L]}B_{l})\Pi'U=U\Sigma U'U=U\Sigma\Rightarrow U=\Pi(\rho\sum_{l\in[L]}B_{l})\Pi'U\Sigma^{-1}.
\end{align*}
Since $\Pi(\mathcal{I},:)=I_{K\times K}$, we have $U(\mathcal{I},:)=(\Pi(\rho\sum_{l\in[L]}B_{l})\Pi'U\Sigma^{-1})(\mathcal{I},:)=\Pi(\mathcal{I},:)(\rho\sum_{l\in[L]}B_{l})\Pi'U\Sigma^{-1}=(\rho\sum_{l\in[L]}B_{l})\Pi'U\Sigma^{-1}$, i.e., $U(\mathcal{I},:)=(\rho\sum_{l\in[L]}B_{l})\Pi'U\Sigma^{-1}$. Therefore, $U=\Pi U(\mathcal{I},:)$ holds. Similarly, we have $V=\Pi V(\mathcal{I},:)$.
\end{proof}
\subsection{Proof of Theorem \ref{mainSPSum}}
\begin{proof}
The following lemma bounds $\|A_{\mathrm{sum}}-\Omega_{\mathrm{sum}}\|_{\infty}$.
\begin{lem}\label{boundAumInfinity}
Under $\mathrm{MLMMSB}(\Pi,\rho,\mathcal{B})$, when Assumption \ref{Assum1} holds, with probability at least $1-o(\frac{1}{n+L})$, we have
\begin{align*}
\|A_{\mathrm{sum}}-\Omega_{\mathrm{sum}}\|_{\infty}=O(\sqrt{\rho nL\mathrm{log}(n+L)}).
\end{align*}
\end{lem}
\begin{proof}
Recall that $\|A_{\mathrm{sum}}-\Omega_{\mathrm{sum}}\|_{\infty}=\mathrm{max}_{i\in[n]}\sum_{j\in[n]}|A_{\mathrm{sum}}(i,j)-\Omega_{\mathrm{sum}}(i,j)|=\mathrm{max}_{i\in[n]}\sum_{j\in[n]}|\sum_{l\in[L]}(A_{l}(i,j)-\Omega_{l}(i,j))|$, next we bound it by using the Bernstein inequality below.
\begin{thm}\label{Bern}
(Theorem 1.4 of \cite{tropp2012user}) Let $\{X_{i}\}$ be independent, random, self-adjoint matrices with dimension $d$. When $\mathbb{E}[X_{i}]=0$ and $\mathrm{and~}\|X_{i}\|\leq R$ almost surely. Then, for all $t\geq 0$,
\begin{align*}
\mathbb{P}(\|\sum_{i}X_{i}\|\geq t)\leq d\cdot \mathrm{exp}(\frac{-t^{2}/2}{\sigma^{2}+Rt/3}),
\end{align*}
where $\sigma^{2}:=\|\sum_{i}\mathbb{E}[X^{2}_{i}]\|$ and $\|\cdot\|$ denotes spectral norm.
\end{thm}

Let $x$ be any $n\times1$ vector. Set $y_{(ij)}=\sum_{l\in[L]}(A_{l}(i,j)-\Omega_{l}(i,j))$ for $i\in[n],j\in[n]$, and $T_{(i)}=\sum_{j\in[n]}y_{(ij)}x(j)$ for $i\in[n]$, we have:
\begin{itemize}
  \item $\mathbb{E}(y_{(ij)}x(j))=0$ for $i\in[n], j\in[n]$.
  \item $|y_{(ij)}x(j)|\leq \tau\mathrm{max}_{j}|x(j)|=\tau\|x\|_{\infty}$ for $i\in[n], j\in[n]$.
  \item Let $\mathrm{Var}(X)$ denote the variance of any random variable $X$. Combine  $A_{l}\in\{0,1\}^{n\times n}$ for $l\in[L]$ with the fact that  $A_{l}(i,m)$ and $A_{l}(m,j)$ are independent when $j\neq i$, we have
\begin{align*}
\sum_{j\in[n]}\mathbb{E}[y^{2}_{(ij)}x^{2}(j)]&=\sum_{j\in[n]}x^{2}(j)\mathbb{E}[y^{2}_{(ij)}]=\sum_{j\in[n]}x^{2}(j)\mathrm{Var}(y_{(ij)})=\sum_{j\in[n]}x^{2}(j)\sum_{l\in[L]}\mathrm{Var}(A_{l}(i,j)-\Omega_{l}(i,j))\\
&=\sum_{j\in[n]}x^{2}(j)\sum_{l\in[L]}\mathrm{Var}(A_{l}(i,j))=\sum_{j\in[n]}x^{2}(j)\sum_{l\in[L]}\Omega_{l}(i,j)(1-\Omega_{l}(i,j))\leq\sum_{j\in[n]}x^{2}(j)\sum_{l\in[L]}\Omega_{l}(i,j)\\
&\leq \sum_{j\in[n]}x^{2}(j)\sum_{l\in[L]}\rho=\rho L\|x\|^{2}_{F}
\end{align*}
\end{itemize}
By Theorem \ref{Bern}, for any $t\geq0$, we have
\begin{align*}
\mathbb{P}(|T_{(i)}|\geq t)\leq\mathrm{exp}(\frac{-t^{2}/2}{\rho L\|x\|^{2}_{F}+\frac{\tau\|x\|_{\infty}t}{3}}).
\end{align*}
Set $t=\frac{\alpha+1+\sqrt{(\alpha+1)(\alpha+19)}}{3}\sqrt{\rho L\|x\|^{2}_{F}\mathrm{log}(n+L)}$ for any $\alpha\geq0$, if $\rho L\|x\|^{2}_{F}\geq\tau^{2}\|x\|^{2}_{\infty}\mathrm{log}(n+L)$, we have
\begin{align*}
\mathbb{P}(|T_{(i)}|\geq t)\leq\mathrm{exp}(-(\alpha+1)\mathrm{log}(n+L)\frac{1}{\frac{18}{(\sqrt{\alpha+1}+\sqrt{\alpha+19})^{2}}+\frac{2\sqrt{\alpha+1}}{\sqrt{\alpha+1}+\sqrt{\alpha+19}}\sqrt{\frac{\tau^{2}\|x\|^{2}_{\infty}\mathrm{log}(n+L)}{\rho L\|x\|^{2}_{F}}}})\leq\frac{1}{(n+L)^{\alpha+1}}.
\end{align*}
Set $x\in\{-1,1\}^{n\times1}$, we have: when $\rho nL\geq\tau^{2}\mathrm{log}(n+L)$, with probability at least $1-o(\frac{1}{(n+L)^{\alpha+1}})$ for any $\alpha\geq0$,
\begin{align*}
T_{(i)}\leq \frac{\alpha+1+\sqrt{(\alpha+1)(\alpha+19)}}{3}\sqrt{\rho nL\mathrm{log}(n+L)}.
\end{align*}
Set $\alpha=1$, when $\rho nL\geq\tau^{2}\mathrm{log}(n+L)$, with probability at least $1-o(\frac{1}{n+L})$, we have
\begin{align*}
\|A_{\mathrm{sum}}-\Omega_{\mathrm{sum}}\|_{\infty}=\mathrm{max}_{i\in[n]}T_{(i)}=O(\sqrt{\rho nL\mathrm{log}(n+L)}).
\end{align*}
\end{proof}
By Theorem 4.2 of \citep{cape2019the}, when $|\lambda_{K}(\Omega_{\mathrm{sum}})|\geq 4\|A_{\mathrm{sum}}-\Omega_{\mathrm{sum}}\|_{\infty}$, there is an orthogonal matrix $\mathcal{O}$ such that
\begin{align*}
\|\hat{U}-U\mathcal{O}\|_{2\rightarrow\infty}\leq14\frac{\|A_{\mathrm{sum}}-\Omega_{\mathrm{sum}}\|_{\infty}\|U\|_{2\rightarrow\infty}}{|\lambda_{K}(\Omega_{\mathrm{sum}})|}.
\end{align*}
Since $\varpi:=\|\hat{U}\hat{U}'-UU'\|_{2\rightarrow\infty}\leq2\|\hat{U}-U\mathcal{O}\|_{2\rightarrow\infty}$, we get
\begin{align*}
\varpi\leq28\frac{\|A_{\mathrm{sum}}-\Omega_{\mathrm{sum}}\|_{\infty}\|U\|_{2\rightarrow\infty}}{|\lambda_{K}(\Omega_{\mathrm{sum}})|}.
\end{align*}
Since $\|U\|_{2\rightarrow\infty}=O(\sqrt{\frac{K}{n}})=O(\sqrt{\frac{1}{n}})$ by Lemma 3.1 of \citep{mao2021estimating} and Condition \ref{condition}, we get
\begin{align*}
\varpi=O(\frac{\|A_{\mathrm{sum}}-\Omega_{\mathrm{sum}}\|_{\infty}}{|\lambda_{K}(\Omega_{\mathrm{sum}})|\sqrt{n}}).
\end{align*}
For $|\lambda_{K}(\Omega_{\mathrm{sum}})|$, by Condition \ref{condition} and Assumption \ref{Assum22}, we have
\begin{align*}
|\lambda_{K}(\Omega_{\mathrm{sum}})|=|\lambda_{K}(\rho\Pi(\sum_{l\in[L]}B_{l})\Pi')|=\rho|\lambda_{K}(\Pi(\sum_{l\in[L]}B_{l})\Pi')|=\rho|\lambda_{K}(\Pi'\Pi\sum_{l\in[L]}B_{l})|\geq\rho\lambda_{K}(\Pi'\Pi)|\lambda_{K}(\sum_{l\in[L]}B_{l})|=O(\rho\frac{n}{K}L)=O(\rho nL),
\end{align*}
which gives that
\begin{align*}
\varpi=O(\frac{\|A_{\mathrm{sum}}-\Omega_{\mathrm{sum}}\|_{\infty}}{\rho n^{1.5}L}).
\end{align*}
By the proof of Theorem 3.2 of \citep{mao2021estimating}, there is a $K$-by-$K$ permutation matrix $\mathcal{P}$ such that,
\begin{align*}
\mathrm{max}_{i\in[n]}\|e'_{i}(\hat{\Pi}-\Pi\mathcal{P})\|_{1}=O(\varpi\kappa(\Pi'\Pi)\sqrt{\lambda_{1}(\Pi'\Pi)})=O(\varpi\sqrt{\frac{n}{K}})=O(\varpi\sqrt{n})=O(\frac{\|A_{\mathrm{sum}}-\Omega_{\mathrm{sum}}\|_{\infty}}{\rho nL}).
\end{align*}
By Lemma \ref{boundAumInfinity}, this theorem holds. Finally, recall that when we use Theorem 4.2 of \citep{cape2019the}, we require $|\lambda_{K}(\Omega_{\mathrm{sum}})|\geq4\|A_{\mathrm{sum}}-\Omega_{\mathrm{sum}}\|_{\infty}$. Since $|\lambda_{K}(\Omega_{\mathrm{sum}})|=O(\rho nL)$, it is easy to see that as long as $\rho nL\gg\|A_{\mathrm{sum}}-\Omega_{\mathrm{sum}}\|_{\infty}$, this requirement holds. The condition $\rho nL\gg\|A_{\mathrm{sum}}-\Omega_{\mathrm{sum}}\|_{\infty}$ holds naturally because we need the row-wise error bound $O(\frac{\|A_{\mathrm{sum}}-\Omega_{\mathrm{sum}}\|_{\infty}}{\rho nL})$ to be much smaller than 1.
\end{proof}
\subsection{An alternative proof of Theorem \ref{mainSPSum}}
Here, we provide an alternative proof of Theorem \ref{mainSPSum} by using Theorem 4.2. of \cite{chen2021spectral}.
\begin{proof}
Set $H_{\hat{U}}=\hat{U}'U$ and let $H_{\hat{U}}=U_{H_{\hat{U}}}\Sigma_{H_{\hat{U}}}V'_{H_{\hat{U}}}$ be its top $K$ singular value decomposition. Define $\mathrm{sgn}(H_{\hat{U}})$ as $U_{H_{\hat{U}}}V'_{H_{\hat{U}}}$. Under $\mathrm{MLMMSB}(\Pi,\rho,\mathcal{B})$, the following results are true.
\begin{itemize}
  \item $\mathbb{E}[\sum_{l\in[L]}(A_{l}(i,j)-\Omega_{l}(i,j))]=0$ for $i\in[n],j\in[n]$.
  \item $\mathbb{E}[(\sum_{l\in[L]}(A_{l}(i,j)-\Omega_{l}(i,j)))^{2}]=\mathbb{E}[\sum_{l\in[L]}(A_{l}(i,j)-\Omega_{l}(i,j))^{2}+\sum_{l_{1}\neq l_{2}, l_{1}\in[L],l_{2}\in[L]}(A_{l_{1}}(i,j)-\Omega_{l_{1}}(i,j))(A_{l_{2}}(i,j)-\Omega_{l_{2}}(i,j))]=\sum_{l\in[L]}\mathbb{E}[(A_{l}(i,j)-\Omega_{l}(i,j))^{2}]=\sum_{l\in[L]}\Omega_{l}(i,j)(1-\Omega_{l}(i,j))\leq\sum_{l\in[L]}\Omega_{l}(i,j)\leq\sum_{l\in[L]}\rho=\rho L$ for $i\in[n],j\in[n]$.
  \item $|A_{\mathrm{sum}}(i,j)-\Omega_{\mathrm{sum}}(i,j)|=|\sum_{l\in[L]}(A_{l}(i,j)-\Omega_{l}(i,j))|\leq \tau$ for $i\in[n],j\in[n]$.
  \item Set $\mu=\frac{n\|U\|^{2}_{2\rightarrow\infty}}{K}$. By Lemma 3.1 of \citep{mao2021estimating}, we have $\frac{1}{K\lambda_{1}(\Pi'\Pi)}\leq\|U\|^{2}_{2\rightarrow\infty}\leq\frac{1}{\lambda_{K}(\Pi'\Pi)}$, which gives that $\mu=O(1)$ by Condition \ref{condition}.
  \item Set $c_{b}=\frac{\tau}{\sqrt{\rho L n/(\mu \mathrm{log}(n))}}$. $\mu=O(1)$ gives $c_{b}=O(\sqrt{\frac{\tau^{2}\mathrm{log}(n)}{\rho nL}})\leq O(1)$ by Assumption \ref{Assum1}.
\end{itemize}
The above results ensure that conditions in Assumption 4.1. \cite{chen2021spectral} are satisfied. Then, by Theorem 4.2. \cite{chen2021spectral}, when $|\lambda_{K}(\Omega_{\mathrm{sum}})|\gg\sqrt{\rho n L\mathrm{log}(n)}$, with probability at least $1-O(\frac{1}{n^{5}})$, we have
\begin{align*}
\|\hat{U}\mathrm{sgn}(H_{\hat{U}})-U\|_{2\rightarrow\infty}=O(\frac{\kappa(\Omega_{\mathrm{sum}})\sqrt{\rho L\mu K}+\sqrt{K\rho L\mathrm{log}(n)}}{|\lambda_{K}(\Omega_{\mathrm{sum}})|}).
\end{align*}
Since $\mu=O(1)$ and $K=O(1)$, we have
\begin{align*}
\|\hat{U}\mathrm{sgn}(H_{\hat{U}})-U\|_{2\rightarrow\infty}=O(\frac{\kappa(\Omega_{\mathrm{sum}})\sqrt{\rho L}+\sqrt{\rho L\mathrm{log}(n)}}{|\lambda_{K}(\Omega_{\mathrm{sum}})|}).
\end{align*}
By Assumption \ref{Assum11} and Condition \ref{condition}, we have  $|\lambda_{K}(\Omega_{\mathrm{sum}})|=O(\rho nL)$ and $|\lambda_{1}(\Omega_{\mathrm{sum}})|=\|\Omega_{\mathrm{sum}}\|=\rho\|\Pi(\sum_{l\in[L]}B_{l})\Pi'\|\leq\rho\|\Pi'\Pi\|\|\sum_{l\in[L]}B_{l}\|=O(\frac{\rho n L}{K})=O(\rho nL)$, which gives that $\kappa(\Omega_{\mathrm{sum}})=O(1)$ and
\begin{align*}
\|\hat{U}\mathrm{sgn}(H_{\hat{U}})-U\|_{2\rightarrow\infty}=O(\frac{(\sqrt{\rho L}+\sqrt{\rho L\mathrm{log}(n)}}{\rho nL})=O(\frac{1}{n}\sqrt{\frac{\mathrm{log}(n)}{\rho L}}).
\end{align*}
Since $\varpi=\|\hat{U}\hat{U}'-UU'\|_{2\rightarrow\infty}\leq2\|U-\hat{U}\mathrm{sgn}(H_{\hat{U}})\|_{2\rightarrow\infty}$, we have
\begin{align*}	
\varpi=O(\frac{1}{n}\sqrt{\frac{\mathrm{log}(n)}{\rho L}}).
\end{align*}
By the proof of Theorem 3.2 in \citep{mao2021estimating}, there is a permutation matrix $\mathcal{P}$ such that,
\begin{align*}
\mathrm{max}_{i\in[n]}\|e'_{i}(\hat{\Pi}-\Pi\mathcal{P})\|_{1}=O(\varpi\kappa(\Pi'\Pi)\sqrt{\lambda_{1}(\Pi'\Pi)})=O(\varpi\sqrt{\frac{n}{K}})=O(\varpi\sqrt{n})=O(\sqrt{\frac{\mathrm{log}(n)}{\rho nL}}).
\end{align*}
Recall that $|\lambda_{K}(\Omega_{\mathrm{sum}})|\geq O(\rho nL)$ and we need $|\lambda_{K}(\Omega_{\mathrm{sum}})|\gg\sqrt{\rho nL\mathrm{log}(n)}$ hold. It is easy to see that as long as $O(\rho nL)\gg\sqrt{\rho nL\mathrm{log}(n)}\Leftrightarrow\rho nL\gg\mathrm{log}(n)$, $|\lambda_{K}(\Omega_{\mathrm{sum}})|\gg\sqrt{\rho nL\mathrm{log}(n)}$ always holds. The condition $\rho nL\gg\mathrm{log}(n)$ holds naturally since we require the row-wise error bound $O(\sqrt{\frac{\mathrm{log}(n)}{\rho nL}})$ to be much smaller than 1.
\end{proof}
\subsection{Proof of Theorem \ref{mainSPDSoS}}
\begin{proof}
Since $\mathbb{E}[S_{\mathrm{sum}}]\neq \tilde{S}_{\mathrm{sum}}$, i.e., $\mathbb{E}[S_{\mathrm{sum}}(i,j)-\tilde{S}_{\mathrm{sum}}(i,j)]\neq0$ for $i\in[n],j\in[n]$, we can not use Theorem 4.2 in \citep{chen2021spectral} to obtain the row-wise eigenspace error $\|\hat{V}\hat{V}'-VV'\|_{2\rightarrow\infty}$ for $S_{\mathrm{sum}}$ and $\tilde{S}_{\mathrm{sum}}$. To obtain SPDSoS's error rate, we use Theorem 4.2 in\citep{cape2019the}. By Bernstein inequality, we have the following lemma that bounds $\|S_{\mathrm{sum}}-\tilde{S}_{\mathrm{sum}}\|_{\infty}$.
\begin{lem}\label{boundSumInfinity}
Under $\mathrm{MLMMSB}(\Pi,\rho,\mathcal{B})$, when Assumption \ref{Assum2} holds, with probability at least $1-o(\frac{1}{n+L})$, we have
\begin{align*}
\|S_{\mathrm{sum}}-\tilde{S}_{\mathrm{sum}}\|_{\infty}=O(\sqrt{\rho^{2}n^{2}L\mathrm{log}(n+L)})+O(\rho^{2}nL).
\end{align*}
\end{lem}
\begin{proof}
Since  $A_{l}\in\{0,1\}^{n\times n}$, we have $A^{2}_{l}(i,j)=A_{l}(i,j)$ for $i\in[n],j\in[n],l\in[L]$, which gives that
\begin{align*}
 \|S_{\mathrm{sum}}-\tilde{S}_{\mathrm{sum}}\|_{\infty}&=\mathrm{max}_{i\in[n]}\sum_{j\in[n]}|S_{\mathrm{sum}}(i,j)-\tilde{S}_{\mathrm{sum}}(i,j)|=\mathrm{max}_{i\in[n]}\sum_{j\in[n]}|\sum_{l\in[L]}(A^{2}_{l}-D_{l}-\Omega^{2}_{l})(i,j)|\\
 &=\mathrm{max}_{i\in[n]}\sum_{j\in[n]}|\sum_{l\in[L]}\sum_{m\in[n]}(A_{l}(i,m)A_{l}(m,j)-\Omega_{l}(i,m)\Omega_{l}(m,j))-\sum_{l\in[L]}D_{l}(i,j)|\\
 &=\mathrm{max}_{i\in[n]}(\sum_{j\neq i,j\in[n]}|\sum_{l\in[L]}\sum_{m\in[n]}(A_{l}(i,m)A_{l}(m,j)-\Omega_{l}(i,m)\Omega_{l}(m,j))|+|\sum_{l\in[L]}\sum_{m\in[n]}(A^{2}_{l}(i,m)-\Omega^{2}_{l}(i,m))-\sum_{l\in[L]}D_{l}(i,i)|)\\
 &=\mathrm{max}_{i\in[n]}(\sum_{j\neq i,j\in[n]}|\sum_{l\in[L]}\sum_{m\in[n]}(A_{l}(i,m)A_{l}(m,j)-\Omega_{l}(i,m)\Omega_{l}(m,j))|+|\sum_{l\in[L]}\sum_{m\in[n]}(A_{l}(i,m)-\Omega^{2}_{l}(i,m))-\sum_{l\in[L]}D_{l}(i,i)|)\\
 &=\mathrm{max}_{i\in[n]}(\sum_{j\neq i,j\in[n]}|\sum_{l\in[L]}\sum_{m\in[n]}(A_{l}(i,m)A_{l}(m,j)-\Omega_{l}(i,m)\Omega_{l}(m,j))|+\sum_{l\in[L]}\sum_{m\in[n]}\Omega^{2}_{l}(i,m))\\
 &\leq \mathrm{max}_{i\in[n]}(\sum_{j\neq i,j\in[n]}|\sum_{l\in[L]}\sum_{m\in[n]}(A_{l}(i,m)A_{l}(m,j)-\Omega_{l}(i,m)\Omega_{l}(m,j))|+\sum_{l\in[L]}\sum_{m\in[n]}\rho^{2}\\
 &=\mathrm{max}_{i\in[n]}(\sum_{j\neq i,j\in[n]}|\sum_{l\in[L]}\sum_{m\in[n]}(A_{l}(i,m)A_{l}(m,j)-\Omega_{l}(i,m)\Omega_{l}(m,j))|+\rho^{2}nL.
\end{align*}
Next, we bound $\sum_{j\neq i,j\in[n]}|\sum_{l\in[L]}\sum_{m\in[n]}(A_{l}(i,m)A_{l}(m,j)-\Omega_{l}(i,m)\Omega_{l}(m,j))|$ for $i\in[n]$. Let $\tilde{x}$ be any $(n-1)\times1$ vector. Set $\tilde{y}_{(ij)}=\sum_{l\in[L]}\sum_{m\in[n]}(A_{l}(i,m)A_{l}(m,j)-\Omega_{l}(i,m)\Omega_{l}(m,j))$ for $i\in[n],j\neq i,j\in[n]$ and $\tilde{T}_{(i)}=\sum_{j\neq i,j\in[n]}\tilde{y}_{(ij)}\tilde{x}(j)$ for $i\in[n]$. The following results hold:
\begin{itemize}
  \item $\mathbb{E}(\tilde{y}_{(ij)}\tilde{x}(j))=0$ since $A_{l}(i,m)$ and $A_{l}(m,j)$ are independent when $j\neq i$ for $i\in[n], j\in[n]$.
  \item $|\tilde{y}_{(ij)}\tilde{x}(j)|\leq \tilde{\tau}\mathrm{max}_{j}|\tilde{x}(j)|=\tilde{\tau}\|\tilde{x}\|_{\infty}$ for $i\in[n], j\in[n]$.
  \item Combine $A_{l}\in\{0,1\}^{n\times n}$ for $l\in[L]$ with the fact that $A_{l}(i,m)$ and $A_{l}(m,j)$ are independent when $j\neq i$, we have
\begin{align*}
\sum_{j\neq i,j\in[n]}\mathbb{E}[\tilde{y}^{2}_{(ij)}\tilde{x}^{2}(j)]&=\sum_{j\neq i,j\in[n]}\tilde{x}^{2}(j)\mathbb{E}[\tilde{y}^{2}_{(ij)}]=\sum_{j\neq i,j\in[n]}\tilde{x}^{2}(j)\mathrm{Var}(\tilde{y}_{(ij)})\\
&=\sum_{j\neq i,j\in[n]}\tilde{x}^{2}(j)\sum_{l\in[L]}\sum_{m\in[n]}\mathrm{Var}(A_{l}(i,m)A_{l}(m,j)-\Omega_{l}(i,m)\Omega_{l}(m,j))=\sum_{j\neq i,j\in[n]}\tilde{x}^{2}(j)\sum_{l\in[L]}\sum_{m\in[n]}\mathrm{Var}(A_{l}(i,m)A_{l}(m,j))\\
&=\sum_{j\neq i,j\in[n]}\tilde{x}^{2}(j)\sum_{l\in[L]}\sum_{m\in[n]}\mathbb{E}[(A_{l}(i,m)A_{l}(m,j)-\Omega_{l}(i,m)\Omega_{l}(m,j))^{2}]\\
&=\sum_{j\neq i,j\in[n]}\tilde{x}^{2}(j)\sum_{l\in[L]}\sum_{m\in[n]}\mathbb{E}[A^{2}_{l}(i,m)A^{2}_{l}(m,j)+\Omega^{2}_{l}(i,m)\Omega^{2}_{l}(m,j)-2A_{l}(i,m)A_{l}(m,j)\Omega_{l}(i,m)\Omega_{l}(m,j)]\\
&=\sum_{j\neq i,j\in[n]}\tilde{x}^{2}(j)\sum_{l\in[L]}\sum_{m\in[n]}(\mathbb{E}[A^{2}_{l}(i,m)A^{2}_{l}(m,j)]-\Omega^{2}_{l}(i,m)\Omega^{2}_{l}(m,j))\\
&=\sum_{j\neq i,j\in[n]}\tilde{x}^{2}(j)\sum_{l\in[L]}\sum_{m\in[n]}(\mathbb{E}[A_{l}(i,m)]\mathbb{E}[A_{l}(m,j)]-\Omega^{2}_{l}(i,m)\Omega^{2}_{l}(m,j))\\
&=\sum_{j\neq i,j\in[n]}\tilde{x}^{2}(j)\sum_{l\in[L]}\sum_{m\in[n]}\Omega_{l}(i,m)\Omega_{l}(m,j)(1-\Omega_{l}(i,m)\Omega_{l}(m,j))\\
&\leq\sum_{j\neq i,j\in[n]}\tilde{x}^{2}(j)\sum_{l\in[L]}\sum_{m\in[n]}\Omega_{l}(i,m)\Omega_{l}(m,j)\leq\sum_{j\neq i,j\in[n]}\tilde{x}^{2}(j)\sum_{l\in[L]}\sum_{m\in[n]}\rho^{2}=\rho^{2}nL\|\tilde{x}\|^{2}_{F}.
\end{align*}
\end{itemize}
By Theorem \ref{Bern}, for any $\tilde{t}\geq0$, we have
\begin{align*}
\mathbb{P}(|\tilde{T}_{(i)}|\geq \tilde{t})\leq\mathrm{exp}(\frac{-\tilde{t}^{2}}{\rho^{2}nL\|\tilde{x}\|^{2}_{F}+\frac{\tilde{\tau}\|\tilde{x}\|_{\infty}\tilde{t}}{3}}).
\end{align*}
Set $\tilde{t}=\frac{\alpha+1+\sqrt{(\alpha+1)(\alpha+19)}}{3}\sqrt{\rho^{2}nL\|\tilde{x}\|^{2}_{F}\mathrm{log}(n+L)}$ for any $\alpha\geq0$. If $\rho^{2}nL\|\tilde{x}\|^{2}_{F}\geq\tilde{\tau}^{2}\|\tilde{x}\|^{2}_{\infty}\mathrm{log}(n+L)$, we have
\begin{align*}
\mathbb{P}(|\tilde{T}_{(i)}|\geq \tilde{t})\leq\mathrm{exp}(-(\alpha+1)\mathrm{log}(n+L)\frac{1}{\frac{18}{(\sqrt{\alpha+1}+\sqrt{\alpha+19})^{2}}+\frac{2\sqrt{\alpha+1}}{\sqrt{\alpha+1}+\sqrt{\alpha+19}}\sqrt{\frac{\tilde{\tau}^{2}\|\tilde{x}\|^{2}_{\infty}\mathrm{log}(n+L)}{\rho^{2}nL\|\tilde{x}\|^{2}_{F}}}})\leq\frac{1}{(n+L)^{\alpha+1}}.
\end{align*}
Recall that $\tilde{x}$ is any $(n-1)\times1$ vector, setting  $\tilde{x}\in\{-1,1\}^{(n-1)\times1}$ gives the following result: when $\rho^{2}n^{2}L\geq\rho^{2}n(n-1)L\geq\tilde{\tau}^{2}\mathrm{log}(n+L)$, with probability at least $1-o(\frac{1}{(n+L)^{\alpha+1}})$ for any $\alpha\geq0$, we have
\begin{align*}
\tilde{T}_{(i)}\leq \frac{\alpha+1+\sqrt{(\alpha+1)(\alpha+19)}}{3}\sqrt{\rho^{2}n(n-1)L\mathrm{log}(n+L)}.
\end{align*}
Set $\alpha=1$, when $\rho^{2}n^{2}L\geq\tilde{\tau}^{2}\mathrm{log}(n+L)$, with probability at least $1-o(\frac{1}{n+L})$, we have
\begin{align*}
\mathrm{max}_{i\in[n]}\tilde{T}_{(i)}=O(\sqrt{\rho^{2}n^{2}L\mathrm{log}(n+L)}).
\end{align*}
Hence, we have
\begin{align*}
\|S_{\mathrm{sum}}-\tilde{S}_{\mathrm{sum}}\|_{\infty}=O(\sqrt{\rho^{2}n^{2}L\mathrm{log}(n+L)})+O(\rho^{2}nL).
\end{align*}
\end{proof}
By Theorem 4.2 of \citep{cape2019the}, if $|\lambda_{K}(\tilde{S}_{\mathrm{sum}})|\geq 4\|S_{\mathrm{sum}}-\tilde{S}_{\mathrm{sum}}\|_{\infty}$, there is an orthogonal matrix $\tilde{\mathcal{O}}$ such that
\begin{align*}
\|\hat{V}-V\tilde{\mathcal{O}}\|_{2\rightarrow\infty}\leq14\frac{\|S_{\mathrm{sum}}-\tilde{S}_{\mathrm{sum}}\|_{\infty}\|V\|_{2\rightarrow\infty}}{|\lambda_{K}(\tilde{S}_{\mathrm{sum}})|}.
\end{align*}
Since $\tilde{\varpi}:=\|\hat{V}\hat{V}'-VV'\|_{2\rightarrow\infty}\leq2\|\hat{V}-V\tilde{\mathcal{O}}\|_{2\rightarrow\infty}$ by basic algebra, we have
\begin{align*}
\tilde{\varpi}\leq28\frac{\|S_{\mathrm{sum}}-\tilde{S}_{\mathrm{sum}}\|_{\infty}\|V\|_{2\rightarrow\infty}}{|\lambda_{K}(\tilde{S}_{\mathrm{sum}})|}.
\end{align*}
Since $\|V\|_{2\rightarrow\infty}=O(\sqrt{\frac{1}{n}})$ by Lemma 3.1 \citep{mao2021estimating} and Condition \ref{condition}, we have
\begin{align*}
\tilde{\varpi}=O(\frac{\|S_{\mathrm{sum}}-\tilde{S}_{\mathrm{sum}}\|_{\infty}}{|\lambda_{K}(\tilde{S}_{\mathrm{sum}})|\sqrt{n}}).
\end{align*}
For $|\lambda_{K}(\tilde{S}_{\mathrm{sum}})|$, by Condition \ref{condition} and Assumption \ref{Assum22}, we have
\begin{align*}
|\lambda_{K}(\tilde{S}_{\mathrm{sum}})|&=\sqrt{\lambda_{K}((\sum_{l\in[L]}\Omega^{2}_{l})^{2})}=\sqrt{\lambda_{K}((\sum_{l\in[L]}\rho^{2}\Pi B_{l}\Pi'\Pi B_{l}\Pi')^{2})}=\rho^{2}\sqrt{\lambda^{2}_{K}(\sum_{l\in[L]}\Pi B_{l}\Pi' \Pi B_{l}\Pi')}\\
&=\rho^{2}\sqrt{\lambda^{2}_{K}(\Pi(\sum_{l\in[L]}B_{l}\Pi'\Pi B_{l})\Pi')}=\rho^{2}\sqrt{\lambda^{2}_{K}(\Pi'\Pi(\sum_{l\in[L]}B_{l}\Pi'\Pi B_{l}))}\geq\rho^{2}\lambda_{K}(\Pi'\Pi)\sqrt{\lambda^{2}_{K}(\sum_{l\in[L]}B_{l}\Pi'\Pi B_{l})}\\
&=O(\rho^{2}\lambda^{2}_{K}(\Pi'\Pi)|\lambda_{K}(\sum_{l\in[L]}B^{2}_{l})|)=O(\rho^{2}n^{2}L),
\end{align*}
which gives that
\begin{align*}
\tilde{\varpi}=O(\frac{\|S_{\mathrm{sum}}-\tilde{S}_{\mathrm{sum}}\|_{\infty}}{\rho^{2}n^{2.5}L}).
\end{align*}
Proof of Theorem 3.2 in \citep{mao2021estimating} gives that, there is a $K\times K$ permutation matrix $\tilde{\mathcal{P}}$ such that,
\begin{align*}
\mathrm{max}_{i\in[n]}\|e'_{i}(\hat{\Pi}-\Pi\tilde{\mathcal{P}})\|_{1}=O(\tilde{\varpi}\kappa(\Pi'\Pi)\sqrt{\lambda_{1}(\Pi'\Pi)})=O(\tilde{\varpi}\sqrt{\frac{n}{K}})=O(\tilde{\varpi}\sqrt{n})=O(\frac{\|S_{\mathrm{sum}}-\tilde{S}_{\mathrm{sum}}\|_{\infty}}{\rho^{2}n^{2}L}).
\end{align*}
By Lemma \ref{boundSumInfinity}, this theorem holds. Finally, since $|\lambda_{K}(\tilde{S}_{\mathrm{sum}})|=O(\rho^{2}n^{2}L)$, the requirement  $|\lambda_{K}(\tilde{S}_{\mathrm{sum}})|\geq4\|S_{\mathrm{sum}}-\tilde{S}_{\mathrm{sum}}\|_{\infty}$ is satisfied as long as $\rho^{2}n^{2}L\gg\|S_{\mathrm{sum}}-\tilde{S}_{\mathrm{sum}}\|_{\infty}$ which holds naturally since we need the row-wise error bound $O(\frac{\|S_{\mathrm{sum}}-\tilde{S}_{\mathrm{sum}}\|_{\infty}}{\rho^{2}n^{2}L})$ to be much smaller than 1.
\end{proof}
\subsection{Proof of Theorem \ref{mainSPSoS}}
\begin{proof}
Lemma \ref{boundSum2Infinity} bounds $\|\sum_{l\in[L]}A^{2}_{l}-\sum_{l\in[L]}\Omega^{2}_{l}\|_{\infty}$.
\begin{lem}\label{boundSum2Infinity}
Under $\mathrm{MLMMSB}(\Pi,\rho,\mathcal{B})$, if $\rho nL\geq\mathrm{log}(n+L)$, with probability at least $1-o(\frac{1}{n+L})$, we have
\begin{align*}
\|\sum_{l\in[L]}A^{2}_{l}-\sum_{l\in[L]}\Omega^{2}_{l}\|_{\infty}=\|S_{\mathrm{sum}}-\tilde{S}_{\mathrm{sum}}\|_{\infty}+O(\rho nL).
\end{align*}
\end{lem}
\begin{proof}
Since $\sum_{l\in[L]}A^{2}_{l}-\sum_{l\in[L]}\Omega^{2}_{l}=S_{\mathrm{sum}}-\tilde{S}_{\mathrm{sum}}+\sum_{l\in[L]}D_{l}$ and each $D_{l}$ is a diagonal matrix for $l\in[L]$, we have
\begin{align*}
\|\sum_{l\in[L]}A^{2}_{l}-\sum_{l\in[L]}\Omega^{2}_{l}\|_{\infty}&=\|S_{\mathrm{sum}}-\tilde{S}_{\mathrm{sum}}+\sum_{l\in[L]}D_{l}\|_{\infty}=\mathrm{max}_{i\in[n]}\sum_{j\in[n]}|S_{\mathrm{sum}}(i,j)-\tilde{S}_{\mathrm{sum}}(i,j)+\sum_{l\in[L]}D_{l}(i,j)|\\
&\leq\mathrm{max}_{i\in[n]}\sum_{j\in[n]}|S_{\mathrm{sum}}(i,j)-\tilde{S}_{\mathrm{sum}}(i,j)|+\mathrm{max}_{i\in[n]}\sum_{j\in[n]}|\sum_{l\in[L]}D_{l}(i,j)|\\
&=\|S_{\mathrm{sum}}-\tilde{S}_{\mathrm{sum}}\|_{\infty}+\mathrm{max}_{i\in[n]}\sum_{l\in[L]}\sum_{j\in[n]}D_{l}(i,j)=\|S_{\mathrm{sum}}-\tilde{S}_{\mathrm{sum}}\|_{\infty}+\mathrm{max}_{i\in[n]}\sum_{l\in[L]}D_{l}(i,i).
\end{align*}
Since Lemma \ref{boundSumInfinity} provides an upper bound of $\|S_{\mathrm{sum}}-\tilde{S}_{\mathrm{sum}}\|_{\infty}$, we only need to bound $\mathrm{max}_{i\in[n]}\sum_{l\in[L]}D_{l}(i,i)$. To bound it, we let $W_{(i)}=\sum_{l\in[L]}D_{l}(i,i)-\sum_{l\in[L]}\sum_{j\in[n]}\Omega_{l}(i,j)=\sum_{l\in[L]}\sum_{j\in[n]}(A_{l}(i,j)-\Omega_{l}(i,j))$ for $i\in[n]$, The following results hold:
\begin{itemize}
  \item $\mathbb{E}[(A_{l}(i,j)-\Omega_{l}(i,j))]=0$ for $l\in[L],i\in[n], j\in[n]$.
 \item $|A_{l}(i,j)-\Omega_{l}(i,j)|\leq1$ for $l\in[L],i\in[n], j\in[n]$.
 \item Since $\mathbb{E}[A_{l}(i,j)]=\Omega_{l}(i,j)$ and $\mathrm{Var}(A_{l}(i,j))=\Omega_{l}(i,j)(1-\Omega_{l}(i,j))$ for Bernoulli distribution, we have
  \begin{align*}
         \sum_{l\in[L]}\sum_{j\in[n]}\mathbb{E}[(A_{l}(i,j)-\Omega_{l}(i,j))^{2}]=\sum_{l\in[L]}\sum_{j\in[n]}\mathrm{Var}(A_{l}(i,j))=\sum_{l\in[L]}\sum_{j\in[n]}\Omega_{l}(i,j)(1-\Omega_{l}(i,j))\leq \sum_{l\in[L]}\sum_{j\in[n]}\Omega_{l}(i,j)\leq\sum_{l\in[L]}\sum_{j\in[n]}\rho=\rho nL.
       \end{align*}
\end{itemize}
By Theorem \ref{Bern}, for any $\tilde{\tilde{t}}\geq0$, we have
\begin{align*}
\mathbb{P}(|W_{(i)}|\geq \tilde{\tilde{t}})\leq\mathrm{exp}(\frac{-\tilde{\tilde{t}}^{2}}{\rho nL+\frac{\tilde{\tilde{t}}}{3}}).
\end{align*}
Set $\tilde{\tilde{t}}=\frac{\alpha+1+\sqrt{(\alpha+1)(\alpha+19)}}{3}\sqrt{\rho nL\mathrm{log}(n+L)}$ for any $\alpha\geq0$, if $\rho nL\geq\mathrm{log}(n+L)$, we have
\begin{align*}
\mathbb{P}(|W_{(i)}|\geq \tilde{\tilde{t}})\leq\mathrm{exp}(-(\alpha+1)\mathrm{log}(n+L)\frac{1}{\frac{18}{(\sqrt{\alpha+1}+\sqrt{\alpha+19})^{2}}+\frac{2\sqrt{\alpha+1}}{\sqrt{\alpha+1}+\sqrt{\alpha+19}}\sqrt{\frac{\mathrm{log}(n+L)}{\rho nL}}})\leq\frac{1}{(n+L)^{\alpha+1}}.
\end{align*}
Hence, when $\rho nL\geq\mathrm{log}(n+L)$, with probability at least $1-o(\frac{1}{(n+L)^{\alpha+1}})$, we have
\begin{align*}
|W_{(i)}|=|\sum_{l\in[L]}D_{l}(i,i)-\sum_{l\in[L]}\sum_{j\in[n]}\Omega_{l}(i,j)|\leq \tilde{\tilde{t}}.
\end{align*}
Let $\alpha=1$, when $\rho nL\geq\mathrm{log}(n+L)$, with probability at least $1-o(\frac{1}{n+L})$, we have
\begin{align*}
\mathrm{max}_{i\in[n]}|W_{(i)}|\leq\frac{2+2\sqrt{10}}{3}\sqrt{\rho nL\mathrm{log}(n+L)}.
\end{align*}
$|\sum_{l\in[L]}D_{l}(i,i)-\sum_{l\in[L]}\sum_{j\in[n]}\Omega_{l}(i,j)|\leq \frac{2+2\sqrt{10}}{3}\sqrt{\rho nL\mathrm{log}(n+L)}$ for $i\in[n]$ gives $\sum_{l\in[L]}D_{l}(i,i)\leq \sum_{l\in[L]}\sum_{j\in[n]}\Omega_{l}(i,j)+\frac{2+2\sqrt{10}}{3}\sqrt{\rho nL\mathrm{log}(n+L)}\leq \rho nL+\frac{2+2\sqrt{10}}{3}\sqrt{\rho nL\mathrm{log}(n+L)}=O(\rho nL)$ since we require $\rho nL\geq\mathrm{log}(n+L)$ to hold. Therefore, we have $\mathrm{max}_{i\in[n]}\sum_{l\in[L]}D_{l}(i,i)=O(\rho nL)$ and this lemma holds.
\end{proof}
Follow a similar proof as Theorem \ref{mainSPDSoS}, for the SPSoS method, we have
\begin{align*}
\mathrm{max}_{i\in[n]}\|e'_{i}(\hat{\Pi}-\Pi\mathcal{P})\|_{1}=O(\frac{\|\sum_{l\in[L]}A^{2}_{l}-\sum_{l\in[L]}\Omega^{2}_{l}\|_{\infty}}{\rho^{2}n^{2}L}).
\end{align*}
By Lemma \ref{boundSum2Infinity} and the fact that $\rho\leq1$, this theorem holds.
\end{proof}
\bibliographystyle{model5-names}\biboptions{authoryear}
\bibliography{refMLMMSB}
\end{document}